\documentclass[a4paper,oneside,reqno,12pt]{amsart}
\usepackage{amssymb}
\usepackage[foot]{amsaddr}
\usepackage{tikz}
\usetikzlibrary{matrix,arrows,decorations.pathmorphing}
\usepackage{caption}
\usepackage{subcaption}
\usepackage{hyperref}

\numberwithin{equation}{section}
\def\Nbbd{\mathbb{N}}
\def\Pbbd{\mathbb{P}}

\def\Rbbd{\mathbb{R}}
\def\Sbbd{\mathbb{S}}
\def\Zbbd{\mathbb{Z}}

\def\Fcal{\mathcal{F}}

\def\Lcal{\mathcal{L}}

\def\Wcal{\mathcal{W}}

\def\mbf{\boldsymbol{m}}

\def\ubf{\boldsymbol{u}}
\def\vbf{\boldsymbol{v}}
\def\wbf{\boldsymbol{w}}
\def\Wbf{\boldsymbol{W}}
\def\xbf{\boldsymbol{x}}
\def\Xbf{\boldsymbol{X}}

\def\xibf{\boldsymbol{\xi}}

\newcommand\sfrak{\mathfrak{s}}
\newcommand\Sfrak{\mathfrak{S}}
\newcommand\tfrak{\mathfrak{t}}
\newcommand\Pfrak{\mathfrak{P}}
\newcommand\Hbf{\mathbf{H}}
\newcommand\Bbf{\mathbf{B}}
\newcommand{\rd}{\mathrm{d}}            
\newcommand{\re}{\mathrm{e}}            
\newcommand{\ri}{\mathrm{i}}            

\def\b{\beta}
\def\g{\gamma}
\def\d{\delta}
\def\eps{\varepsilon}
\def\la{\lambda}
\def\s{\sigma}
\def\phi{\varphi}
\def\om{\omega}

\newcommand\sgn{\mathop{\hbox{\rm sgn}}\nolimits}

\newcommand\abs[1]{\left|#1\right|}
\def\dd{\partial}
\def\Ref#1{(\ref{#1})}

\def\endxmpl{\hfill$\bigcirc$}
\def\?{(?)\marginpar[\hfill?|]{|?}}

\def\beq{\begin{equation}}
\def\eeq{\end{equation}}
\def\be{\begin{equation*}}
\def\ee{\end{equation*}}
\newcounter{problem}
{\refstepcounter{problem}{\noindent\bf Problem \theproblem.\marginpar[{\hfill\maltese}]{{\maltese}}}\em }{}
\newcounter{hypothesis}
{\refstepcounter{hypothesis}{\noindent\bf Conjecture \thehypothesis.\marginpar[{\hfill$\blacktriangleleft$}]{{$\blacktriangleleft$}}}\em }{}

\newtheorem{conj}{Conjecture}

\newtheorem{prop}{Proposition}
\newtheorem{lemma}{Lemma}

\newcounter{myexample}

\newcommand\ket[1]{\left|\,#1\right\rangle}

\newcommand\angles[1]{\left\langle#1\right\rangle}
\newcommand\Psid{\Psi^\dagger}
\newcommand\norm[1]{\left\|#1\right\|}
\newcommand\Pequiv{\stackrel{\Pfrak}{\equiv}}
\newcommand\twored{\stackrel{2}{\equiv}}
\newcommand\threered{\stackrel{3}{\equiv}}

\begin{document}
\title{Quantisation of Kadomtsev-Petviashvili equation}
\date{\today}
\author{K K Kozlowski$^1$}
\address{$^1$Univ Lyon, Ens de Lyon, Univ Claude Bernard, CNRS, Laboratoire de Physique, F-69342 Lyon, France.}
\author{E Sklyanin$^2$}
\address{$^2$Department of Mathematics, University of York, York YO10 5DD, UK}
\author{A Torrielli$^3$}
\address{$^3$Department of Mathematics, University of Surrey, Guildford, GU2 7XH, UK}

\begin{abstract}
A quantisation of the KP equation on a cylinder is proposed
that is equivalent to an infinite system of non-relativistic one-dimensional bosons
carrying masses $m=1,2,\ldots$ The Hamiltonian is Galilei-invariant
and includes the split $\Psid_{m_1}\Psid_{m_2}\Psi_{m_1+m_2}$
and merge $\Psid_{m_1+m_2}\Psi_{m_1}\Psi_{m_2}$ terms
for all combinations of particles with masses $m_1$, $m_2$ and $m_1+m_2$,
with a special choice of coupling constants.
The Bethe eigenfunctions for the model are constructed.
The consistency of the coordinate Bethe Ansatz, and therefore,
the quantum integrability of the model is verified up to the mass $M=8$ sector.
\end{abstract}
\maketitle

\section{Introduction}\label{sec:intro}

The Kadomtsev-Petviashvili (KP) equation \cite{konopelchenko1993introduction}
\beq\label{eq:KP_eqmotion}  
   \phi_{t\s} - \phi_{xx} - 2\b(\phi\phi_\s)_\s  + \g\phi_{\s\s\s\s} = 0,
\eeq
is one of the most studied nonlinear integrable equations in 2+1 variables $(\s,x;t)$.
The aim of the present paper is to construct a quantised version of KP
while preserving its integrability.

\textit{A warning:} For reasons explained below, we have deliberately deviated from the standard notation of \cite{konopelchenko1993introduction}
by changing the conventional variable $x$ to $\s$ and $y$ to $x$.

Traditionally, KP is considered as an equation in (2+1)-dimensional space-time, both variables
$\s$ and $x$ playing the role of spatial variables. For our purposes, however, we take a different stance,
viewing only $x$ as a genuine spatial variable and downgrading $\s$ to a mere label indexing the continuum
of fields $\phi$ in (1+1)-dimensional space-time. The notation $(\s,x)$ stresses the changed
roles of the two variables.

We also choose $x$ to run from $-\infty$ to $\infty$, whereas imposing the periodicity condition
$\s\equiv\s+2\pi$ on $\s$. We assume that $\phi\rightarrow0$ sufficiently fast as $x\rightarrow\pm\infty$.

We have also introduced two real coupling constants $\b$ and $\g$ into the equation.
Though, in the classical case, they can be removed by a rescaling of the variables $x$, $\s$, $\phi$,
they are useful for the quantisation and for discussing the limiting cases.
Note that the constants $\b$ and $\g$
may have arbitrary sign, the case $\g>0$ corresponding to the so-called KP-I, respectively
$\g<0$ to KP-II, and $\g=0$ to the so-called dispersionless KP \cite{ManakovSantini2008}.
Since the substitution $\phi:=-\phi$ results in changing the sign of $\b$, one may assume that $\b\ge0$.

The paper is organised as follows.
In Section \ref{sec:classical_kp}, we describe the Poisson structure and the Hamiltonian
of the classical model.
In Section \ref{sec:quantisation}, we quantise the model using the simplest
normal ordering prescription for the Hamiltonian. Passing from the field
$\phi(\s,x)$ to its Fourier components in the variable $\s$ we obtain the description
of the system in terms of the discrete infinite set of canonical fields $\Psid_m(x)$,
$\Psi_m(x)$ labelled by the index $m=1,2,3,\ldots$ and describing scalar nonrelativistic
bosons of mass $m$. The Hamiltonian is Galilei-invariant
and includes the split $\Psid_{m_1}\Psid_{m_2}\Psi_{m_1+m_2}$
and merge $\Psid_{m_1+m_2}\Psi_{m_1}\Psi_{m_2}$ terms
for all combinations of particles with masses $m_1$, $m_2$ and $m_1+m_2$,
with a special choice of coupling constants.

In Section \ref{sec:fock_space},
we describe the Fock space $\Fcal$ of the system and introduce a convenient notation
to handle the infinite number of fields. We realise as well the action of the Hamiltonian on the $N$-particle
state as a differential operator with singular delta-function coefficients.

Due to the conservation of the total mass $M$, the quantum-field-theoretic problem
is reduced to a sequence of quantum-mechanical  problems in sectors
of fixed mass $M$. The structure of mass-$M$
sector $\Fcal_M$ is analysed in Section \ref{sec:structure_of_mass-M_sector}.
Since the number of particles is not preserved,
the sector $\Fcal_M$ splits into the
orthogonal sum of subspaces $\Fcal_M^{\,\mbf}$ labelled by compositions $\mbf$ of number $M$.
The corresponding wave functions are defined on Weyl alcoves $x_1<x_2<\ldots<x_N$,
where $N$ is the length of $\mbf=(m_1,m_2,\ldots,m_N)$.

In Section \ref{sec:from_delta_to_bc},
we interprete the delta-function terms in the Hamiltonian as jump conditions
for the derivatives of components of the wave function, and formulate the complete set of differential
equations and boundary conditions for the wave functions.

In Section \ref{sec:M2sln},
we solve the eigenvalue problem in the sector $M=2$, and compute the two-particle $S$-matrix,
as a rational function having 3 poles and 3 zeroes. The two possible arrangements of the poles
are labelled as quantum KP-I and KP-II cases.

In Section \ref{sec:BAm=111}, we formulate the Bethe Ansatz
in the subsector $\Fcal_M^{(1\ldots1)}$ containing only particles of mass-1.
The Bethe eigenfunction is written as a linear combination of plain waves
with the coefficients that reproduce the correct 2-particle $S$-matrices.
In Section \ref{sec:BAgeneric}, we extend the Bethe Ansatz to the generic sector
of particles with different masses, and formulate the factorisation conjecture
that allows one to reduce the verification of the consistency equations to those for the
subsector $\Fcal_M^{(M)}$ containing a single particle of mass $M$.
In Section \ref{sec:BAm=(M)}, we analyse those equations and describe a solution
that is verified by means of computer algebra up to $M\leq8$.
A more technical discussion of the involved combinatorial issues is left for the appendices.

In the concluding Section \ref{sec:discussion} we sum up the results and discuss the unsolved questions
and perspectives.

\section{Classical KP}\label{sec:classical_kp} 
In this paper, we use the following notation:
$\Zbbd=\{0,\pm1,\pm2,\ldots\}$ stands for the set of integers, $\Nbbd=\{1,2,3,\ldots\}$ the set of natural numbers,
$\Nbbd_0=\{0,1,2,3,\ldots\}$ the set of non-negative integers, $\Rbbd$ the set of real numbers,
$\Sbbd^1=\Rbbd/2\pi\Zbbd$ a circle.

The classical Kadomtsev-Petviashvili (KP) integrable hierarchy \cite{konopelchenko1993introduction}
is formulated in terms of a real-valued scalar field $\phi(\s,x)$ on the cylinder
$\Sbbd^1\times\Rbbd$. The field $\phi$ vanishes sufficiently fast as $x\rightarrow\pm\infty$ and has Poisson brackets
\beq\label{phi-pb}
    \{\phi(\s,x),\phi(\tau,y)\}=2\pi\d'(\s-\tau)\d(x-y),\qquad
    \s,\tau\in \Sbbd^1,\quad x,y\in\Rbbd.
\eeq

Due to the periodicity of $\phi(\s,x)$ in $\s$, the average of $\phi$ over $\Sbbd^1$
belongs to the center of the bracket \Ref{phi-pb}.
In what follows we always set it to 0, assuming that
\beq\label{phi0average}
    \int_0^{2\pi} \rd\s\,\phi(\s,x)=0\quad \forall x\in\Rbbd.
\eeq

Due to \Ref{phi0average}, the antiderivative $\dd_\s^{-1}$ on the space of functions
with zero average over $\mathbb{S}^1$ is defined correctly (one can always choose the integration constant
in a unique way).

There exists an infinite series of commuting Hamiltonians $H_p$, $p=0,1,2,\ldots$
\beq
     \{H_p,H_q\}=0
\eeq
expressed as integrals of local (w.r.t. $x$) densities
\beq
    H_p=\int_{0}^{2\pi} \frac{\rd\s}{2\pi}\int_{-\infty}^{\infty} \rd x\,
    h_p(\s,x).
\eeq
\beq
    h_p(\s,x)=\frac12\,\bigl(\dd_\s^{-p}\phi\bigr)\bigl(\dd_x^p\phi\bigr)+O(\b)+O(\g),\qquad
    \b,\g\rightarrow0
\eeq
such that
\begin{subequations}\label{eq:classical_hamiltonian_densities}
\begin{align}
    h_0(\s,x)&=\frac12\,\phi^2(\s,x),\\
    h_1(\s,x)&=\frac12\,\bigl(\dd_\s^{-1}\phi\bigr)\bigl(\dd_x\phi\bigr),  \\
    h_2(\s,x)&=\frac12\,\bigl(\dd_\s^{-2}\phi\bigr)\bigl(\dd_x^2\phi\bigr)
            +\frac{\b}{3}\phi^3+\frac{\g}{2}(\dd_\s\phi)^2.
\end{align}
\end{subequations}

The corresponding equations of motion $\dd_{t_p}=\{\cdot,H_p\}$
are
\begin{subequations}
\begin{align}
     \phi_{t_0}&=\phi_\s,\\
     \phi_{t_1}&=-\phi_x,\\
     \phi_{t_2}&=\dd_\s^{-1}\phi_{xx}+2\b\phi\phi_\s-\g\phi_{\s\s\s}. \label{eqmo2}
\end{align}
\end{subequations}

Note that $H_0$ and $H_1$ are generators of translations in $\s$ and $x$ respectively.

Differentiating \Ref{eqmo2} in respect to $\s$ we obtain the KP equation in the form
\Ref{eq:KP_eqmotion}.

Note that the equations of motion \Ref{eq:KP_eqmotion} are invariant under the Galilei
transform \newline $x:=x+2vt$, $\s:=\s+vx+v^2t$.
The infinitesimal Galilei boost
\beq
    B=\int_{0}^{2\pi} \frac{\rd\s}{2\pi}\int_{-\infty}^{\infty} \rd x\,
      x h_0(\s,x),\qquad
    \{\phi,B\}=x\phi_\s
\eeq
commutes with the Hamiltonians as follows:
\beq
    \{H_p,B\}=-pH_{p-1}.
\eeq

\section{Quantisation}\label{sec:quantisation} 

Using the correspondence principle $[\cdot,\cdot]\simeq \ri\hbar\{\cdot,\cdot\}$
and setting $\hbar=1$ we obtain from \Ref{phi-pb} the commutation relations for the field $\phi(\s,x)$:
\beq\label{comm-phi}
   [\phi(\s,x),\phi(\tau,y)]=2\pi\ri\d'(\s-\tau)\d(x-y),\qquad \phi^\dagger=\phi.
\eeq

The corresponding Fourier components
\beq
    \phi(\s,x)=\sum_{n\in\Zbbd} a_n(x)\,\re^{-\ri n\s},\qquad
    a_n(x)=\int_0^{2\pi}\frac{\rd\s}{2\pi}\,\phi(\s,x)\,\re^{\ri n\s},\qquad
    a_0(x)=0,
\eeq
form the Heisenberg (oscillator) Lie algebra
\cite{kac1987bombay}
\beq\label{heisenberg algebra}
    [a_m(x),a_n(y)]=m\d_{m+n,0}\d(x-y),\qquad
    a_n^\dagger(x)=a_{-n}(x).
\eeq

Consider the highest-weight (h.w.) module generated
by the h.w.\ vector (vacuum) $\ket{0}$ such that
\beq
    a_n(x)\ket{0}=0,\qquad n>0.
\eeq

Equivalently, the h.w.\ module is isomorphic to the bosonic Fock space $\Fcal$ generated by the
canonical creation/annihilation operators $\Psid_n(x)$ and $\Psi_n(x)$
\beq
     \Psi_n(x)=n^{-1/2} a_n(x),\qquad \Psi_n^\dagger(x)=n^{-1/2} a_{-n}(x),\qquad n\in\Nbbd,\quad x\in\Rbbd,
\eeq
\beq\label{commPsi}
[\Psi_m (x),\Psid_n(y)]=\delta_{mn}\delta(x-y), \quad \Psi_m(x)\ket{0}=0,
\qquad m,n\in\Nbbd, \quad x,y\in\Rbbd.
\eeq

Our quantisation prescription for the Hamiltonians $H_0$, $H_1$ and $H_2$
is to take the classical expressions \Ref{eq:classical_hamiltonian_densities},
replace $\phi$ with the quantum operators
and apply the Wick normal ordering: $\Psi^\dagger$ to the left, $\Psi$ to the right.
The result is
\begin{align}
   \Hbf_0&=\int_{0}^{2\pi} \frac{\rd\s}{2\pi}\int_{-\infty}^{\infty} \rd x\,
        :\frac12\,\phi^2(\s,x):\
      =\sum_{m\in\Nbbd}\,m\,\int_{-\infty}^\infty \rd x\, \Psi_m^\dagger(x)\Psi_m(x), \label{H0-Psi}\\
   \Hbf_1&=\int_{0}^{2\pi} \frac{\rd\s}{2\pi}\int_{-\infty}^{\infty} \rd x\,
        :\frac12\,\bigl(\dd_\s^{-1}\phi\bigr)\bigl(\dd_x\phi\bigr):\
      =-\ri\sum_{m\in\Nbbd}\,\int_{-\infty}^\infty \rd x\, \Psi_m^\dagger(x)\,\dd_x\Psi_m(x), \label{H1-Psi} \\
   \Hbf_2&=\int_{0}^{2\pi} \frac{\rd\s}{2\pi}\int_{-\infty}^{\infty} \rd x\,
        :\frac12\,\bigl(\dd_\s^{-2}\phi\bigr)\bigl(\dd_x^2\phi\bigr)+\frac{\b}{3}\phi^3
        +\frac{\g}{2}(\dd_\s\phi)^2: \notag\\
   &=-\sum_{m\in\Nbbd}\,\frac{1}{m}\,\int_{-\infty}^\infty \rd x\, \Psi_m^\dagger(x)\,\dd_x^2\Psi_m(x)\notag\\
   &+\sum_{m_1,m_2\in\Nbbd} \beta_{m_1m_2}\int_{- \infty}^{\infty} \rd x \,
   \bigl[ \Psi_{m_1 + m_2}^\dagger (x) \, \Psi_{m_1} (x) \, \Psi_{m_2} (x)
+ \, \Psi_{m_1}^\dagger (x) \, \Psi_{m_2}^\dagger (x) \, \Psi_{m_1 + m_2} (x) \bigr] \notag\\
  &+\sum_{m\in\Nbbd} \g_m \int_{-\infty}^\infty \rd x\, \Psi_m^\dagger(x)\Psi_m(x), \label{H2-Psi}
\end{align}
where
\beq\label{def-betam1m2gammam}
     \b_{m_1m_2}=\b_{m_2m_1}=\b\,\sqrt{(m_1+m_2)m_1m_2},\qquad
    \g_m=\g m^3.
\eeq

As in the classical case, $\Hbf_0$ and $\Hbf_1$ being, respectively, generators
of $\s$- and $x$-translations, commute between themselves and with $\Hbf_2$.
The quantum Galilei boost
\beq
   \Bbf=\int_{-\infty}^\infty \rd x\,x\,:\frac12\,\phi^2(\s,x):
   =\sum_{m\in\Nbbd}\,m\,\int_{-\infty}^\infty \rd x\, x\Psi_m^\dagger(x)\Psi_m(x)
\eeq
commutes with the Hamiltonians as follows:
\beq
    [\Hbf_0,\Bbf]=0,\qquad
    [\Hbf_1,\Bbf]=-\ri \Hbf_0,\qquad
    [\Hbf_2,\Bbf]=-2\ri \Hbf_1.
\eeq

Physically, the Hamiltonian $\Hbf_2$ describes a non-relativistic, Galilei-invariant
system of one-dimensional Bose-particles labelled by the integer index $m$ that can be interpreted as particle's mass.
The interaction is local. The cubic $\b$-terms describe
processes where 2 particles of masses $m_1$ and $m_2$ merge into one of mass $m_1+m_2$ and the respective splitting.
The unitary transformation $\Psi_m\mapsto-\Psi_m$, $\Psi_m^\dagger\mapsto-\Psi_m^\dagger$ simply changes the sign of
$\b$, so one may assume $\b\ge0$.
For $\b$=0 the fields decouple, and one gets the theory of free particles with masses $m$
and the rest energy $\g m^3$.

A model with such kind of interaction was first proposed in \cite{Lee1954}, and its variants and generalisations
under the general name `Lee model' were popular in 1950-60s as toy models in nuclear physics.
Our variant of the Lee model is distinguished on several counts: first, by being 1D, second, by using infinitely
many fields, and third, by the specific choice of coupling constants
\Ref{def-betam1m2gammam} that, as we are expecting, makes the theory integrable.
Other examples of integrable 1D models of Lee type that have been studied previously
include the $N$-waves model \cite{kulish1986quantum} and continuous magnet \cite{Sklyanin1988}.

The crucial question is thus whether the integrability of the theory is preserved in the quantum case.
One way of checking the integrability would be to construct higher commuting quantum Hamiltonians
$H_n$, $n\geq3$ for which the normal ordering prescription can not be expected to work.
Moreover, the problem of higher local quantum Hamiltonians
is notoriously difficult even in a much simpler case of the quantum nonlinear
Schr\"{o}dinger equation \cite{sklyanin1982quantum}: the higher Hamiltonians are known to be extremely
singular and do not have well-defined normal symbols
\cite{Case1984polynomial,DaviesKorepin1989,Gutkin1985conservation}.

As our test of integrability, we choose instead to construct an explicit formula
for the simultaneous eigenfunctions of
$\Hbf_0$, $\Hbf_1$ and $\Hbf_2$ by means of the coordinate Bethe Ansatz
and to show that the multiparticle $S$-matrices are factorised into 2-particles ones.

\section{Fock space}\label{sec:fock_space}


The canonical operators $\Psid_m(x)$ and $\Psi_m(x)$ are labelled by the pairs $(m,x)\in\Nbbd\times\Rbbd$.

It is convenient to treat the pair of labels as a single composite entity $\xi=(m,x)$, or
$\eta=(n,y)$. Denoting $\d_{\xi\eta}=\delta_{mn}\delta(x-y)$ we can thus rewrite \Ref{commPsi} as
\beq\label{commPsixi}
     [\Psi_\xi,\Psid_\eta]=\d_{\xi\eta},\qquad \Psi_\xi\ket{0}=0.
\eeq

The bosonic Fock space $\Fcal$ is decomposed into $N$-particle components spanned by the vectors
\beq\label{defketf}
    \ket{f}=\sum_{N=0}^\infty \frac{1}{N!}
    \sum_{\mbf\in\Nbbd^N} \int_{\Rbbd^N} \rd x_1\ldots \rd x_N\,
    f_N\binom{\mbf}{\xbf}\,
    \prod_{j=1}^N \Psid_{m_j}(x_j)\ket{0}
\eeq
defined in terms of the $N$-particle wave functions
\beq
        f_N\binom{\mbf}{\xbf}
        =f_N\begin{pmatrix}
             m_1, & m_2, &\ldots& m_N\\
             x_1, & x_2, &\ldots& x_N
            \end{pmatrix}
        =f_N(\xi_1,\ldots,\xi_N)=f_N(\xibf)
\eeq
depending on $N$ discrete indices
$\mbf=(m_1,\ldots,m_N)\in\Nbbd^N$ and $N$ continuous variables \mbox{$\xbf=(x_1,\ldots,x_N)\in\Rbbd^N$},
and symmetric with respect to permutations of the pairs \mbox{$\xi_i=(m_i,x_i)$}.
We shall use the notation
\beq
    (\ket{f})_N(\xibf) \equiv f_N(\xibf)
\eeq
to refer to the $N$-particle component of the vector $\ket{f}$.

Using the shorthand notation, one can rewrite \Ref{defketf} as
\beq\label{ketfxi}
    \ket{f}= \sum_{N=0}^\infty \frac{1}{N!}
    \int \rd\xibf^N f_N(\xibf) \prod_{j=1}^N \Psid_{\xi_j}\ket{0},
\eeq
where
\beq
    \int \rd\xibf^N = \sum_{\mbf\in\Nbbd^N} \int_{\Rbbd^N} \rd x_1\ldots \rd x_N.
\eeq

The norm of the vector $\ket{f}$ is
\beq
   \norm{f}^2=\langle f|\,f\rangle=
   \sum_{N=0}^\infty \frac{1}{N!}
   \int \rd\xibf^N\,
      \abs{f_N(\xibf)}^2.
\eeq

From \Ref{commPsixi} and \Ref{ketfxi} the action of the canonical operators on the $N$-particle wave function
can be computed easily:
\beq\label{eq:Psif2}
    \bigl(\Psi_\eta\,\ket{f}\bigr)_N(\xi_1,\ldots,\xi_N)
    =f_{N+1}(\xi_1,\ldots,\xi_N,\eta),
\eeq
or, simply
\beq\label{eq:Psif3}
    \bigl(\Psi_\eta\,\ket{f}\bigr)_N(\xibf)
    =f_{N+1}(\xibf,\eta).
\eeq

Respectively,
\beq\label{eq:Psidf4}
    \bigl(\Psid_\eta\,\ket{f}\bigr)_N(\xi_1,\ldots,\xi_N)
   = \sum_{j=1}^N \d_{\eta\xi_j}
     f_{N-1}(\xi_1,\ldots,\widehat{\xi_j},\ldots,\xi_N),
\eeq
where $\widehat{\xi_j}$ means omitting $\xi_j$.

From \Ref{eq:Psif3}, and \Ref{eq:Psidf4}
one easily derives the action of the Hamiltonians $\Hbf_0$, $\Hbf_1$ and $\Hbf_2$
on the state $\ket{f}$ in terms of its $N$-particle components.

From \Ref{H0-Psi} one computes that
\beq
    (\Hbf_0\ket{f})_N\binom{\mbf}{\xbf}
       =\abs{\mbf}\, f_N\binom{\mbf}{\xbf},\qquad
     \abs{\mbf}=m_1+\ldots+m_N,
\eeq
measuring thus the total mass of an $N$-particle system.

Similarly, from \Ref{H1-Psi} one derives that $\Hbf_1$ is the total momentum operator
(generator of infinitesimal translation):
\beq
   (\Hbf_1\ket{f})_N\binom{\mbf}{\xbf}
   =-\ri(\dd_{x_1}+\ldots+\dd_{x_N})f_N\binom{\mbf}{\xbf}.
\eeq

From \Ref{H2-Psi} the action of $\Hbf_2$ on $\ket{f}$ takes the form: 
\begin{multline}\label{H_2-N}
\bigl(\Hbf_2\ket{f}\bigr)_N\begin{pmatrix}
             m_1 &\ldots& m_N\\
             x_1 &\ldots& x_N
            \end{pmatrix}
   = -\left(\frac{1}{m_1}\,\dd_{x_1}^2+\ldots+\frac{1}{m_N}\,\dd_{x_N}^2\right)
   f_N\begin{pmatrix}
             m_1 &\ldots& m_N\\
             x_1 &\ldots& x_N
            \end{pmatrix}  \\
   +2\sum_{1\leq i_1<i_2\leq N} \b_{m_{i_1}m_{i_2}}
     f_{N-1}\begin{pmatrix}
             m_1 &\ldots& \widehat{m_{i_1}} &\ldots& \widehat{m_{i_2}}  &\ldots&  m_N, & m_{i_1}+m_{i_2}\\
             x_1 &\ldots& \widehat{x_{i_1}} &\ldots& \widehat{x _{i_2}} &\ldots& x_N, & x_{i_1}
            \end{pmatrix}
            \d(x_{i_1}-x_{i_2}) \\
   +\sum_{k=1}^N \sum_{\substack{n_1,n_2\in\Nbbd \\ n_1+n_2=m_k}} \b_{n_1n_2}
    f_{N+1}\begin{pmatrix}
             m_1 &\ldots& \widehat{m_k} &\ldots& m_N & n_1 & n_2\\
             x_1 &\ldots& \widehat{x_k} &\ldots& x_N & x_k & x_k
            \end{pmatrix} \\
   + (\g_{m_1}+\ldots+\g_{m_N})
   f_N\begin{pmatrix}
             m_1 &\ldots& m_N\\
             x_1 &\ldots& x_N
            \end{pmatrix}.
\end{multline}

As in \Ref{eq:Psidf4}, the hat marks omitted arguments. Due to the symmetry of the wave function,
the order of the arguments is irrelevant, so we put the
new arguments replacing the omitted ones at the end of the list.

Whereas the operators $\Hbf_0$ and $\Hbf_1$ preserve the number $N$ of particles,
the Hamiltonian $\Hbf_2$ does not do so, due to the exchange terms
$\Psi_{m_1 + m_2}^\dagger\Psi_{m_1}\Psi_{m_2}$ and
$\Psi_{m_1}^\dagger\Psi_{m_2}^\dagger\Psi_{m_1 + m_2}$.
However, since $\Hbf_2$ commutes with $\Hbf_0$,
it preserves the mass $M=m_1+\ldots+m_N$ instead.
The original quantum-field-theoretical model
splits thus into a series of quantum-mechanical ones restricted to
the eigenspaces $\Fcal_M$ of $\Hbf_0$ which we call mass-$M$ sectors.
By \Ref{H_2-N}, in each mass-$M$ sector the Hamiltonian $\Hbf_2$ is represented by a multicomponent differential
operator with singular (delta-function) coefficients.

\section{Structure of mass-$M$ sector}\label{sec:structure_of_mass-M_sector}

To describe the structure of the mass-$M$ sector of the Fock space in more details we shall need
a few definitions from combinatorics \cite{StanleyFomin_book}.

A \emph{composition} $\mbf$ of a nonnegative integer $M\in\Nbbd$ is defined  as
a sequence $\mbf=(m_1,\ldots,m_N)$ of $m_i\in\Nbbd$ such that $m_1+\ldots+m_N\equiv\abs\mbf=M$.
The number $N=\ell(\mbf)$ is called \emph{length} of the composition, and $M=\abs{\mbf}$ its
\emph{weight}. The number of compositions of $M$ equals $2^{M-1}$.

We introduce a partial order $\succ$ on the set of compositions: $\mbf\succ\widetilde\mbf$ means that
$\ell(\widetilde\mbf)=\ell(\mbf)-1$ and $\widetilde\mbf$ can be obtained from $\mbf$ by replacing an adjacent
pair $(m_i,m_{i+1})$ for some $i=1,\ldots,\ell(\mbf)-1$ with $m_i+m_{i+1}$, that is
\begin{align*}
    \mbf&=(m_1,\ldots,m_{i-1},m_i,m_{i+1},m_{i+2},\ldots,m_N),\\
    \widetilde\mbf&=(m_1,\ldots,m_{i-1},m_i+m_{i+1},m_{i+2},\ldots,m_N).
\end{align*}

The set of compositions $\mbf$ of $M$ becomes then an ordered graph, with
vertices $\mbf$ and arrows pointing from $\mbf$ to $\widetilde\mbf$ if $\mbf\succ\widetilde\mbf$.
The graph is topologically equivalent to an $(M-1)$-dimensional hypercube
having $2^{M-1}$ vertices and $(M-1)2^{M-2}$ edges, as exemplified by Fig.\ \ref{fig:hypercubes}.
The vertex $(1,\ldots,1)$ is the \textit{source}, having no predecessors.
Starting from it and travelling along the arrows
one can reach any point of the hypercube in a variety of ways, terminating at the
\textit{sink} $(M)$.

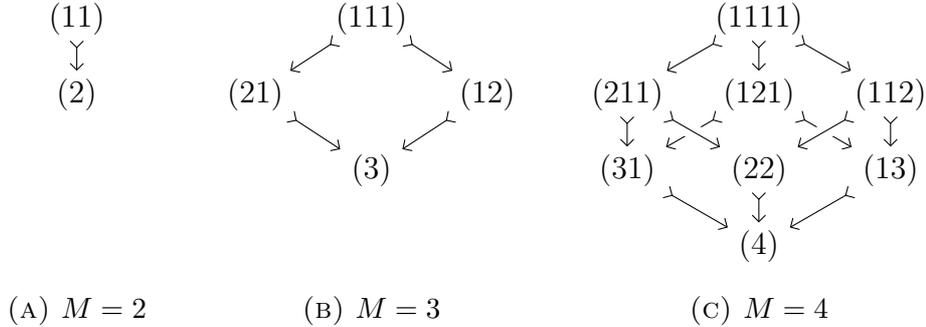
\begin{figure}[h]
\begin{subfigure}[b]{0.15\textwidth}
\centering
\begin{tikzpicture}[>=angle 90]
\matrix (m) [matrix of math nodes,
row sep=1em, column sep=1em,
text height=1.5ex,
text depth=0.25ex]{
(11) \\
(2) \\
\vphantom{(2)} \\
\vphantom{(2)} \\
};
\path[>->]
(m-1-1) edge (m-2-1);
\end{tikzpicture}
\caption{$M=2$}
\label{fig:hypecubes2}
\end{subfigure}
\
\begin{subfigure}[b]{0.3\textwidth}
\centering
\begin{tikzpicture}[>=angle 90]
\matrix (m) [matrix of math nodes,
row sep=1em, column sep=1em,
text height=1.5ex,
text depth=0.25ex]{
& (111) & \\
(21) &  & (12) \\
& (3) & \\
& \vphantom{(2)} & \\
};
\path[>->]
(m-1-2) edge (m-2-1)
(m-1-2) edge (m-2-3)
(m-2-1) edge (m-3-2)
(m-2-3) edge (m-3-2);
\end{tikzpicture}
\caption{$M=3$}
\end{subfigure}
\begin{subfigure}[b]{0.32\textwidth}
\centering
\begin{tikzpicture}[>=angle 90,
back line/.style={densely dotted},
cross line/.style={preaction={draw=white, -,
line width=6pt}}]
\matrix (m) [matrix of math nodes,
row sep=1em, column sep=1em,
text height=1.5ex,
text depth=0.25ex]{
& (1111) & \\
(211) & (121) & (112) \\
(31) & (22) & (13) \\
& (4) & \\
};
\path[>->]
(m-1-2) edge (m-2-1)
        edge (m-2-2)
        edge (m-2-3)
(m-2-1) edge (m-3-1)
(m-2-2) edge (m-3-1)
(m-2-1) edge [cross line] (m-3-2)
(m-2-2) edge (m-3-3)
(m-2-3) edge [cross line] (m-3-2)
(m-2-3) edge (m-3-3)
(m-3-1) edge (m-4-2)
(m-3-2) edge (m-4-2)
(m-3-3) edge (m-4-2)
;
\end{tikzpicture}
\caption{$M=4$}
\end{subfigure}

\caption{Composition hypercubes}
\label{fig:hypercubes}
\end{figure}

The mass-$M$ sector $\Fcal_M$ of our Fock space $\Fcal$ is the eigenspace of the mass operator $\Hbf_0$ corresponding to
the eigenvalue $M$. It is spanned by the vectors
\beq\label{eq:massM-ket}
      \ket{f} = \sum_{\mbf:\ \abs{\mbf}=M} \frac{1}{N!}\int_{\Rbbd^{N}} \rd x_1\ldots\rd x_N\,
    f_{N}\binom{\mbf}{\xbf}\,
    \prod_{j=1}^N \Psid_{m_j}(x_j)\ket{0}\in\Fcal_M,
\eeq
with the norm
\beq\label{norm-f-comp}
    \norm{f}^2= \sum_{\mbf:\ \abs{\mbf}=M} \frac{1}{N!}
    \int_{\Rbbd^{N}} \rd x_1\ldots\rd x_N\,
    \abs{f_{N}\binom{\mbf}{\xbf}}^2
\eeq
(here and below we always imply $N=\ell(\mbf)$).

The \emph{Weyl alcove} $\Wcal_N$ is defined as
\beq\label{eq:Weyl_alcove}
     \Wcal_N = \bigl\{ \xbf\in\Rbbd^N:\ x_1< x_2< \ldots < x_N  \bigr\}.
\eeq

Due to the symmetry of the wave functions, the terms $f_N(\xibf)$ contribute to the sum \Ref{eq:massM-ket}
with the multiplicity $N!$.
Consequently, one can replace the integration over $\Rbbd^N$
in \Ref{eq:massM-ket} and in \Ref{norm-f-comp}
with the integration over $\Wcal_N$, having adjusted the combinatorial coefficients:
\beq\label{eq:massM-ket-weyl}
      \ket{f} = \sum_{\mbf:\ \abs{\mbf}=M}
      \int_{\Wcal_{N}} \rd x_1\ldots\rd x_N\,
    f_{N}\binom{\mbf}{\xbf}\,
    \prod_{j=1}^{N} \Psid_{m_j}(x_j)\ket{0}\in\Fcal_M,
\eeq
\beq\label{norm-f-weyl}
    \norm{f}^2= \sum_{\mbf:\ \abs{\mbf}=M}
    \int_{\Wcal_{N}} \rd x_1\ldots\rd x_N\,
    \abs{f_{N}\binom{\mbf}{\xbf}}^2.
\eeq

As a result, $\Fcal_M$ splits into the orthogonal sum
\beq
     \Fcal_M=\bigoplus_{\mbf:\ \abs{\mbf}=M} \Fcal_M^{\,\mbf},
\eeq
where $\Fcal_M^{\,\mbf}\simeq L^2\big( \Wcal_{N} \big)$.

The vectors of $\Fcal_M$ are thus identified with the collection
of $2^{M-1}$ functions $f_N(\xibf)$ labelled by the compositions $\mbf$,
with arguments $\xbf\in\Wcal_N$. The component $f_N(\xibf)$
describes a collection of $N=\ell(\mbf)$ one-dimensional particles with masses $m_i$ and
coordinates $x_i$ ordered from left to right.

In the next section we shall rewrite the eigenvalue problem for the differential operator \Ref{H_2-N}
with delta-function coefficients on functions $f_N(\xibf)$ with $\xbf\in\Rbbd^N$
as an equivalent system of differential equations and boundary conditions for functions
$f_N(\xibf)$ with \mbox{$\xbf\in\Wcal_N$}.

\section{From $\d$-function to boundary conditions}\label{sec:from_delta_to_bc}  

Replacing a delta-function term with boundary conditions is a standard trick,
see e.g.\ \cite{BogoliubiovIzerginKorepinBookCorrFctAndABA,Gaudin2014BWf,LiebLiniger1963},
for the case of a scalar Bose-gas (quantum nonlinear Schr\"odinger equation).
We only need to adapt the technique to the case of particles of different masses.

Let us analyse first a simple two-particle example. Let a function $f(x_1,x_2)$
on $\Rbbd^2$ satisfy the
Schr\"odinger equation describing two particles of masses $m_1$ and $m_2$
and containing a singular inhomogeneous term (external source)
\beq\label{eq:Ef}
   \left[\left(-\frac{1}{m_1}\,\dd_{x_1}^2-\frac{1}{m_2}\,\dd_{x_2}^2\right)\right]f(x_1,x_2)
   +\s(x_1)\d(x_1-x_2)+\tau(x_1,x_2)=0,
\eeq
where the densities $\s(x)$ and $\tau(x_1,x_2)$ are assumed to be smooth functions.

Since $\d(x_1-x_2)$ vanishes off the diagonal, one obtains immediately the
differential equation ``in the bulk''
\beq\label{eq:bulk2}
    -\left(\frac{1}{m_1}\,\dd^2_{x_1}+\frac{1}{m_2}\,\dd^2_{x_2}\right)f(x_1,x_2)+\tau(x_1,x_2)=0, \qquad
    x_1\neq x_2.
\eeq

To derive the boundary conditions on the diagonal $x_1=x_2$,
assume that the function $f$ is piecewise smooth, meaning that it is
given by two different expressions
$f^{(+)}(x_1,x_2)$ in the half-plane $x_1-x_2>0$ and
$f^{(-)}(x_1,x_2)$ in the half-plane $x_1-x_2<0$. Furthermore, both functions $f^{(\pm)}$
are assumed to be smooth and defined in an open neighbourhood of the cut $x_1=x_2$,
the domain of each function extending thus beyond its native half-plane.

Introducing the step function
\beq
    \theta(x)=\left\{\begin{array}{rcl}
                   1, &\quad& x>0 \\
                   0, &\quad& x<0
               \end{array}\right.
\eeq
one can represent $f$ as
\beq\label{split-f}
    f(x_1,x_2)=f^{(+)}(x_1,x_2)\theta(x_1-x_2)+f^{(-)}(x_1,x_2)\theta(x_2-x_1)
\eeq
(we treat $f$ as a measurable function defining a distibution, so its values on the
zero-measure set $x_1=x_2$ are irrelevant and can be left undefined).

Substitute now \Ref{split-f} into \Ref{eq:Ef} and perform the differentiations, using
the identities valid for any smooth function $\om(x_1,x_2)$
$$
   \om(x_1,x_2)\d(x_1-x_2) = \om(x_1,x_1)\d(x_1-x_2),
$$
$$
   \om(x_1,x_2)\d'(x_1-x_2) = \bigl(\dd_{x_2}\om\bigr)(x_1,x_1)\d(x_1-x_2)+\om(x_1,x_1)\d'(x_1-x_2),
$$
so that the coefficients at $\theta$- and $\d$-functions in the resulting sum depend only on $x_1$.
The coefficients at $\theta(x_1-x_2)$ and $\theta(x_2-x_1)$ then give the ``bulk'' equation
\Ref{eq:bulk2} for $f^{(+)}$ and $f^{(-)}$, respectively. The coefficient at
$\d'(x_1-x_2)$ gives the continuity condition
\beq\label{eq:contfpm}
      f^{(+)}(x_1,x_1) = f^{(-)}(x_1,x_1).
\eeq

The coefficient at $\d(x_1-x_2)$ gives, after a simplification using \Ref{eq:contfpm},
the boundary condition
\beq\label{eq:jumpfpm}
     \left[\frac{1}{m_1}\dd_{x_1}(f^{(+)}-f^{(-)})
     +\frac{1}{m_2}\dd_{x_2}(-f^{(+)}+f^{(-)})\right](x_1,x_1)
     =\s(x_1).
\eeq

Let $g(x\pm0)$ denote the one-sided limiting values for $g(x\pm\eps)$ as $\eps\searrow0$.
Then, one can rewrite the continuity condition \Ref{eq:contfpm} as
\beq
     f(x+0,x-0) = f(x-0,x+0) \equiv f(x,x),
\eeq
and the jump-of-transversal-derivative condition \Ref{eq:jumpfpm} as
\begin{multline}\label{eq:jump2}
     -\bigl[(m_2\,\dd_{x_1}-m_1\,\dd_{x_2}) f \bigr](x+0,x-0)
    +\bigl[(m_2\,\dd_{x_1}-m_1\,\dd_{x_2}) f \bigr](x-0,x+0) \\
    +m_1m_2\,\sigma(x)=0.
\end{multline}

In \Ref{eq:jump2}, it is assumed that the differential operator is applied first to the function
of two variables $(x_1,x_2)$ and then the limit $\eps\searrow0$ is taken in
$(x_1,x_2)=(x\pm \eps,x \mp \eps)$.

The above argument works also in the multiparticle case since
in the neighbourhood of the cut $x_i=x_j$ the functions depend on the rest of the variables
continuously, and we can ignore all remaining variables.
Besides, when $x$'s are ordered as in \Ref{eq:Weyl_alcove}
the only jump conditions to take into account are those for the adjacent particles $x_i=x_{i+1}$.

Consider the eigenvalue problem $\Hbf_2\ket{f} = \la\ket{f}$.
Taking \Ref{H_2-N} and applying \Ref{eq:bulk2} we then obtain the set of bulk equations labelled by the vertices $\mbf$ of the hypercube
\begin{equation}\label{eq:bulk-gen}
\la f_N(\xibf)
   =
   \sum_{i=1}^N \left( -\frac{1}{m_i}\,\dd_{x_i}^2+\g_{m_i} \right)
   f_N(\xibf)   +  \sum_{\mbf'\succ\mbf} \b_{n_1n_2}
    f_{N+1}(\xibf^{\prime}).
\end{equation}

The sum  over the compositions $\mbf'$ such that $\mbf'\succ\mbf$ is in fact a double sum over the integers $k,n_1,n_2$ with $n_1+n_2=m_k$
that label the compositions and the associated vectors $\xibf, \xibf'$ as
\begin{align}
 \xibf & =  \begin{pmatrix} \mbf \\ \xbf  \end{pmatrix} =
	  \begin{pmatrix}
             m_1 &\ldots& m_{k-1} & m_k & m_{k+1} &\ldots& m_N \\
             x_1 &\ldots& x_{k-1} & x_k & x_{k+1} &\ldots& x_N
            \end{pmatrix},  \nonumber \\
  \xibf' & =  \begin{pmatrix} \mbf' \\ \xbf'  \end{pmatrix} =
\begin{pmatrix}
             m_1 &\ldots& m_{k-1} & n_1 & n_2 & m_{k+1} &\ldots& m_N \\
             x_1 &\ldots& x_{k-1} & x_k & x_k & x_{k+1} &\ldots& x_N
             \end{pmatrix}.
\label{compositions for bulk equation}
\end{align}

Note also that in \eqref{eq:bulk-gen} one does not need to distinguish the limits $x_k\pm0$ owing to  the continuity of $f_{N+1}$.

Respectively, \Ref{eq:jump2} produces the set of jump conditions labelled by the
arrows of the hypercube pointing from $\mbf$ that is the pairs
\be
   (m_1,\ldots,m_k,m_{k+1},\ldots,m_N)
   \rightarrow(m_1,\ldots,m_k+m_{k+1},\ldots,m_N),\qquad
   k=1,2,\ldots,N-1.
\ee

Applying \Ref{eq:jump2} to \Ref{H_2-N}
we get the jump of transversal derivative on the line $x_k=x_{k+1}$:
\begin{multline}\label{eq:jump-gen-raw}   
 -\bigl[(m_{k+1}\,\dd_{x_k}-m_k\,\dd_{x_{k+1}})
 f_N\bigr]\begin{pmatrix}
             m_1 &\ldots& m_{k-1} & m_k & m_{k+1} & m_{k+2} & \ldots & m_N\\
             x_1 &\ldots& x_{k-1} & x_{k+1}+0 & x_{k+1}-0 & x_{k+2} & \ldots & x_{N}
            \end{pmatrix} \\
 +\bigl[(m_{k+1}\,\dd_{x_k}-m_k\,\dd_{x_{k+1}})
 f_N\bigr]\begin{pmatrix}
             m_1 &\ldots& m_{k-1} & m_k & m_{k+1} & m_{k+2} & \ldots & m_N\\
             x_1 &\ldots& x_{k-1} & x_k-0 & x_k+0 & x_{k+2} & \ldots & x_{N}
            \end{pmatrix} \\
 +2m_km_{k+1}\b_{m_k,m_{k+1}} f_{N-1}\begin{pmatrix}
             m_1 &\ldots& m_{k-1} & \widehat{m_k} & \widehat{m_{k+1}}& m_{k+2} &\ldots & m_N &  m_k+m_{k+1}\\
             x_1 &\ldots& x_{k-1} & \widehat{x_k} & \widehat{x_{k+1}} & x_{k+2} & \ldots & x_{N} & x_k
            \end{pmatrix} = 0 .
\end{multline}
 To express the result
as a function of the arguments $(x_1<\ldots<\widehat{x_{k+1}}<\ldots< x_N)$
on a Weyl alcove
it remains to use the symmetry of $f_{N}$ and $f_{N-1}$ to swap $x_k\leftrightarrow x_{k+1}$ in the first term of \Ref{eq:jump-gen-raw}
and, respectively, to rearrange the $x$'s in the increasing order  in the last term.
The resulting final form of the jump condition can be recast in a compact form by using the compositions
$\mbf',\mbf''\succ\mbf$ of $\mbf$ and the associated vectors $\xibf', \xibf''$ and $\xibf$:
\begin{align}
   \xibf'&= \begin{pmatrix} \mbf' \\ \xbf' \end{pmatrix} 
	  =\begin{pmatrix}
             m_1 &\ldots& m_k & m_{k+1} & \ldots & m_N\\
             x_1 &\ldots& x_k-0 & x_k+0 & \ldots & x_{N}
            \end{pmatrix}  , \nonumber\\
    \xibf''&= \begin{pmatrix} \mbf'' \\ \xbf'' \end{pmatrix} 
	  =\begin{pmatrix}
             m_1 &\ldots& m_{k+1} & m_{k} & \ldots & m_N\\
             x_1 &\ldots& x_k-0 & x_k+0 & \ldots & x_{N}
            \end{pmatrix} ,  \nonumber \\
    \xibf&= \begin{pmatrix} \mbf \\ \xbf \end{pmatrix} 
	  =\begin{pmatrix}
             m_1 &\ldots& m_{k-1} & m_k+m_{i+1} & m_{k+2} &\ldots & m_N\\
             x_1 &\ldots& x_{k-1} & x_k & x_{k+2} & \ldots & x_{N}
            \end{pmatrix}  .
    \label{introduction compositions eqn saut}
\end{align}

\begin{multline}\label{eq:jump-gen}
\bigl[(m_k\,\dd_{x_k}-m_{k+1}\,\dd_{x_{k+1}}) f_N\bigr](\xibf'')
 +  \bigl[(m_{k+1}\,\dd_{x_k}-m_k\,\dd_{x_{k+1}}) f_N\bigr](\xibf')   \\
 +2m_km_{k+1}\b_{m_k,m_{k+1}} f_{N-1}(\xibf)  = 0.
\end{multline}

Note that the swapping $m_k\leftrightarrow m_{k+1}$ produces an identical equation.
Also, for \mbox{$m_k=m_{k+1}$} the first and the second term in \Ref{eq:jump-gen} are equal.

To conclude, the eigenvalue problem for the Hamiltonian $\Hbf_2$ in the sector of mass $M$
is now formulated in terms of a set of functions $f_N(\xibf)$ labelled by compositions $\mbf$ of $M$
with length $N=\ell(\mbf)$ defined on Weyl alcoves $\Wcal_N$.
The equations for $f_N(\xibf)$ are divided into two classes: the bulk differential equations
of 2nd order \Ref{eq:bulk-gen} labelled by the vertices $\mbf$ of the compositions hypercube,
and the jump conditions \Ref{eq:jump-gen} for the transversal derivatives
labelled by the edges of the compositions hypercube that correspond to the merging
$(m_k,m_{k+1})\rightarrow(m_k+m_{k+1})$ of the adjacent particles.

\section{Solution in the sector $M=2$}\label{sec:M2sln} 

The number $M=2$ admits two compositions: $1+1$ and $2$, see Fig. \ref{fig:hypecubes2}.
Respectively, the mass-2 sector splits as $\Fcal_2=\Fcal_2^{(11)}\oplus\Fcal_2^{(2)}$, so that any vector $\ket{f}\in \Fcal_2$ can be represented as
\beq
     \ket{f} = \int_{x_1<x_2} \rd x_1\rd x_2\,f_2\begin{pmatrix} 1 & 1 \\ x_1 & x_2 \end{pmatrix}
     \Psi_1^\dagger(x_1)\Psi_1^\dagger(x_2)\ket{0}
     +\int_{-\infty}^\infty \rd x_1\,f_1\begin{pmatrix} 2 \\ x_1 \end{pmatrix}
     \Psi_2^\dagger(x_1)\ket{0},
\eeq

\beq
     \norm{f}^2 = \int_{x_1<x_2} \rd x_1\rd x_2\,\left|f_2\begin{pmatrix} 1 & 1 \\ x_1 & x_2 \end{pmatrix}\right|^2
     +\int_{-\infty}^\infty \rd x_1\,\left|f_1\begin{pmatrix} 2 \\ x_1 \end{pmatrix}\right|^2.
\eeq

The general bulk equation \Ref{eq:bulk-gen} produces two bulk equations corresponding to the vertices
of the graph in Fig.\ \ref{fig:hypecubes2}:  
\begin{subequations}\label{M=2problem}
\begin{align}\label{bulk11}
    \la\, f_2\begin{pmatrix} 1 & 1 \\ x_1 & x_2 \end{pmatrix}
    &= (-\dd_{x_1}^2-\dd_{x_2}^2+2\g_1)f_2\begin{pmatrix} 1 & 1 \\ x_1 & x_2 \end{pmatrix},\\
\label{bulk2}
     \la\, f_1\begin{pmatrix} 2 \\ x_1 \end{pmatrix}
     &= \left(-\frac12\dd_{x_1}^2+\g_2\right)f_1\begin{pmatrix} 2 \\ x_1 \end{pmatrix}
     +\b_{11}\, f_2\begin{pmatrix} 1 & 1 \\ x_1 & x_1 \end{pmatrix}.
\end{align}

Respectively, equation \Ref{eq:jump-gen} produces the jump
condition corresponding to the single arrow $(11)\rightarrow(2)$ of the graph in Fig.\ \ref{fig:hypecubes2}:
\beq\label{jump11-2}
    2\bigl[(\dd_{x_1}-\dd_{x_2})f_2\bigr]\begin{pmatrix} 1 & 1 \\ x_1 & x_1 \end{pmatrix}
    +2\b_{11}\,f_1\begin{pmatrix} 2 \\ x_1 \end{pmatrix} = 0
\eeq
(the two terms with derivatives coincide due to the symmetry $m_1=m_2=1$).
\end{subequations}

In the spirit of Bethe Ansatz \cite{BogoliubiovIzerginKorepinBookCorrFctAndABA,Gaudin2014BWf},
we look for a solution of the boundary
problem \Ref{M=2problem} in the subsector $\mbf=(11)$ as a linear combination of plain waves:
the incoming one $\re^{\ri(u_2x_1+u_1x_2)}$ and the scattered one $\re^{\ri(u_1x_1+u_2x_2)}$,
with the scattering coefficient $S_{21}$:
\begin{subequations}\label{M=2ansatz}
\beq\label{M=2ansatzf2}
    f_2\begin{pmatrix} 1 & 1 \\ x_1 & x_2 \end{pmatrix}
    = \re^{\ri(u_2x_1+u_1x_2)}+S_{21}\re^{\ri(u_1x_1+u_2x_2)}, \qquad x_1<x_2.
\eeq

The jump condition \Ref{jump11-2} implies then that the wave function $f_1$
in the subsector $\mbf=(2)$  has to be the exponent $\re^{\ri(u_1+u_2)x_1}$,
up to a coefficient $R$:
\beq\label{M=2ansatzf1}
    f_1\begin{pmatrix} 2 \\ x_1 \end{pmatrix}
    = R\re^{\ri(u_1+u_2)x_1}.
\eeq
\end{subequations}

Substituting the Ansatz \Ref{M=2ansatz} into \Ref{M=2problem}
we obtain, respectively, the bulk-11 equation:
\begin{subequations}\label{M=2algeqs}
\beq\label{M=2lambda}
  u_1^2+u_2^2+2\g_1 = \la,
\eeq
the bulk-2 equation:
\beq
   \left(\frac12(u_1+u_2)^2+\g_2\right)\,R + \b_{11}(1+S_{21}) = \la R,
\eeq
and the jump $(11)\rightarrow(2)$ equation:
\beq
   \ri(u_2-u_1)+\ri(u_1-u_2)S_{21} + \b_{11}R = 0.
\eeq
\end{subequations}

The system of three linear equations \Ref{M=2algeqs} for $\la$, $S_{21}$, $R$
is easily solved. Equation \Ref{M=2lambda} gives immediately the value of $\la$, and the two remaining
equations produce the answer 
\beq\label{def S(u)}
     S_{21}= S(u_2-u_1), \qquad
     S(u) = -\frac{P(\ri u)}{P(-\ri u)}
\eeq
where $P$ is the cubic polynomial
\beq\label{eq:def-Pv}
    P(v) = v^3 + (2\g_2-4\g_1)v-2\b_{11}^2,
\eeq
or substituting $\b_{11}=\sqrt{2}\b$, $\g_1=\g$, $\g_2=8\g$
from \Ref{def-betam1m2gammam},
\beq\label{eq:defP(u)}
    P(v) = v^3 + 12\g\,v - 4\b^2.
\eeq

Respectively,
\beq
     R = \frac{4\ri\b_{11}u_{21}}{P(-\ri u_{21})} = \frac{4\sqrt{2}\ri\b u_{21}}{P(-\ri u_{21})},\qquad
     u_{21} \equiv u_2-u_1.
\eeq

As befits a Galilei-invariant theory, the $S$-matrix is invariant under
the simultaneous translations $u_a \mapsto u_a + c$, $a=1,2$.

Note that the sum of the zeroes of the cubic polynomial $P$ is 0
due to the absence of the quadratic term.
Since $P$ has real coefficients and negative free term $-4\b^2$,
it has exactly one positive root.
The two remaining zeroes lie in the left half-plane, and their
position is determined by the discriminant $D=-432(16\g^3+\b^4)$.
For $D<0$ they are complex-conjugated, and for $D>0$ they are both real negative,
as shown on Fig.\ \ref{fig:zeroes of P}.

It is tempting to call the case $D<0$, or $16\g^3+\b^4>0$ the quantum KP-I
equation, and $D>0$, or $16\g^3+\b^4<0$ the quantum KP-II equation.
Note that the boundary between the two cases is not $\g=0$ as in the classical case
but $16\g^3+\b^4=0$ when $P$ has a double negative zero, the term  $\b^4$ playing the role of a quantum correction.
It remains disputable what to call the \textit{dispersionless} quantum KP: either $16\g^3+\b^4=0$,
or $\g=0$ that corresponds to $P(u)=u^3-4\b^2$, the zeroes forming an equilateral triangle.

The corresponding scattering coefficient $S(u)$ given by \Ref{def S(u)}
is a rational function having three zeroes and three poles.
Their positions ($\circ$ for zeroes, $\bullet$ for poles), depending on $D$, are shown on Fig.\ \ref{fig:zerpolS}

\setlength{\unitlength}{0.2\textwidth}
\begin{figure}[h]
\begin{subfigure}[c]{0.4\textwidth}
\centering
\begin{picture}(2,1)(-1,-0.5)
\put(-1,0){\line(1,0){1.55}}
\put(0.65,0){\line(1,0){0.3}}
\put(0,-0.5){\line(0,1){1}}
\put(0.6,0){\circle{0.1}}
\put(-0.3,0.2){\circle{0.1}}
\put(-0.3,-0.2){\circle{0.1}}
\end{picture}
\caption{qKP-I: $D<0$}
\end{subfigure}
\qquad
\begin{subfigure}[c]{0.4\textwidth}
\centering
\begin{picture}(2,1)(-1,-0.5)
\put(-1,0){\line(1,0){0.55}}
\put(-0.35,0){\line(1,0){0.1}}
\put(-0.15,0){\line(1,0){0.70}}
\put(0.65,0){\line(1,0){0.3}}
\put(0,-0.5){\line(0,1){1}}
\put(0.6,0){\circle{0.1}}
\put(-0.2,0){\circle{0.1}}
\put(-0.4,0){\circle{0.1}}
\end{picture}
\caption{qKP-II: $D>0$}
\end{subfigure}

\caption{Zeroes of $P(u)$}
\label{fig:zeroes of P}
\end{figure}
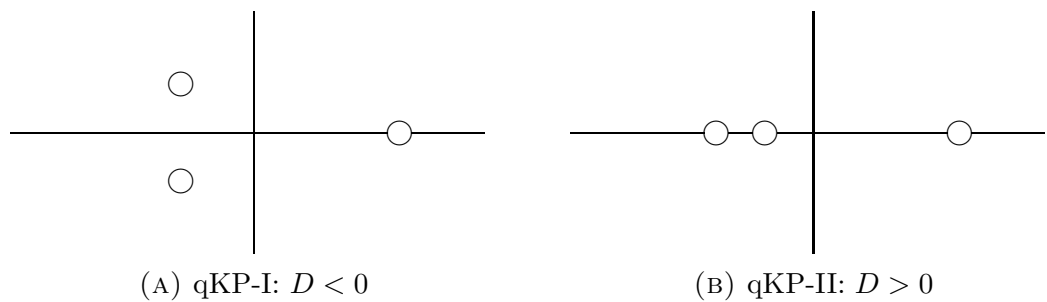

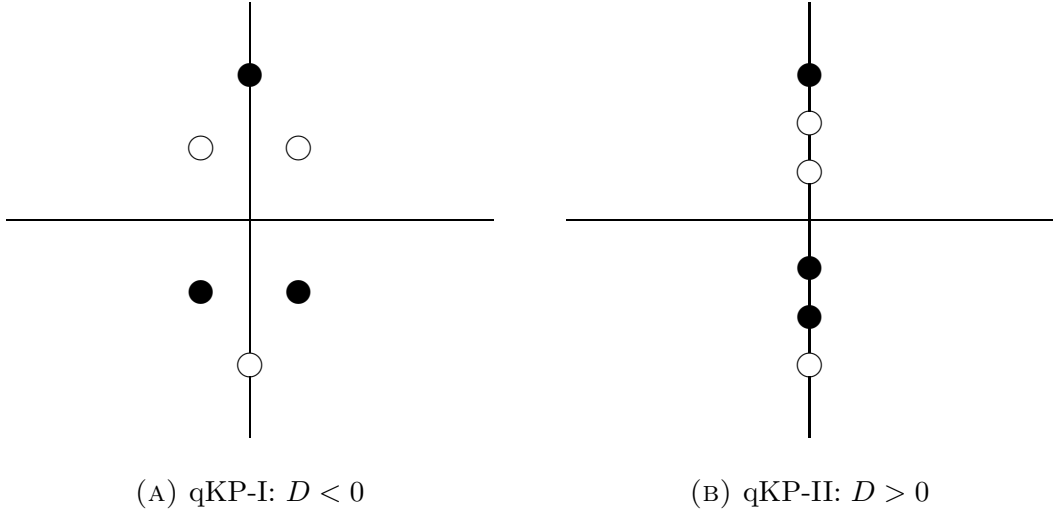
\begin{figure}[h]
\begin{subfigure}[c]{0.4\textwidth}
\centering
\begin{picture}(2,2)(-1,-1)
\put(-1,0){\line(1,0){2}}
\put(0,-0.9){\line(0,1){0.25}}
\put(0,-0.55){\line(0,1){1.45}}
\put(0,0.6){\circle*{0.1}}
\put(0,-0.6){\circle{0.1}}
\put(-0.2,0.3){\circle{0.1}}
\put(0.2,0.3){\circle{0.1}}
\put(-0.2,-0.3){\circle*{0.1}}
\put(0.2,-0.3){\circle*{0.1}}
\end{picture}
\caption{qKP-I: $D<0$}
\label{fig:zeroesP}
\end{subfigure}
\qquad
\begin{subfigure}[c]{0.4\textwidth}
\centering
\begin{picture}(2,2)(-1,-1)
\put(-1,0){\line(1,0){2}}
\put(0,-0.9){\line(0,1){0.25}}
\put(0,-0.55){\line(0,1){0.7}}
\put(0,0.25){\line(0,1){0.1}}
\put(0,0.45){\line(0,1){0.45}}
\put(0,0.6){\circle*{0.1}}
\put(0,-0.6){\circle{0.1}}
\put(0,-0.4){\circle*{0.1}}
\put(0,-0.2){\circle*{0.1}}
\put(0,0.4){\circle{0.1}}
\put(0,0.2){\circle{0.1}}
\end{picture}
\caption{qKP-II: $D>0$}
\end{subfigure}
\caption{Zeroes and poles of $S(u)$}
\label{fig:zerpolS}
\end{figure}

\section{Bethe Ansatz in sector $\Fcal_M^{(1\ldots1)}$}\label{sec:BAm=111} 

The known multiparticle integrable models share two common features: the absence of diffraction
(preservation of the asymptotic momenta of the particles after collision),
and the factorisation of the multiparticle $S$-matrix into
two-particles factors \cite{BogoliubiovIzerginKorepinBookCorrFctAndABA,Gaudin2014BWf}.
For the models with a delta-function interaction, like the quantum nonlinear Schr\"odinger model
\cite{berezin1964schrodinger,LiebLiniger1963}, such behaviour is manifested via
the coordinate Bethe Ansatz \cite{BogoliubiovIzerginKorepinBookCorrFctAndABA,Gaudin2014BWf},
or the assumption that the eigenfunction can be written as a sum of plane waves
with the coefficients differing by a two-particles $S$-factor when any pair of momenta permutes.

In this section, we shall describe the Bethe Ansatz for our model in the subsector
$\Fcal_M^{(1\ldots1)}$ of $\Fcal_M$ containing $M$ particles of unit mass
and corresponding to the composition $\mbf=(1,\ldots,1)$ of $M$.


Let $\Sfrak_{[1,M]}$ be the permutation group of $(1,\ldots,M)$. Let
$\ubf\equiv(u_1,\ldots,u_M)$ be the vector of momenta, and
$\vbf\equiv\ri\ubf$. Define the action of a permutation $\sfrak=(\sfrak_1,\ldots,\sfrak_M)\in\Sfrak_{[1,M]}$
on functions of $\vbf$ by substitutions $\sfrak:\,v_j\mapsto v_{\sfrak_j}$.
Then, for a plane wave
\beq
    \exp(\vbf\cdot\xbf)\equiv \exp\left( v_1 x_1+\ldots+v_M x_M\right)
\eeq
we have
\beq
    \sfrak\bigl(\exp(\vbf\cdot\xbf)\bigr) =
    \bigl(\exp(\sfrak(\vbf)\cdot\xbf)\bigr) =
    \exp\left( v_{\sfrak_1} x_1+\ldots+v_{\sfrak_M} x_M\right).
\eeq

We choose to normalise the Bethe wave function by making the coefficients at the plane waves polynomial
in $v_j$. Such a normalisation was proposed first for the quantum nonlinear Schr\"odinger equation in \cite{Gaudin2014BWf}, \textit{c.f.}\  Chapter 4, eq.\ (4.8),
see also \cite{BogoliubiovIzerginKorepinBookCorrFctAndABA}, Chapter 1, eq.\ (1.24).
Such a choice has the advantage of allowing for
algebraic manipulations with polynomials rather than rational functions.

\begin{conj}[Bethe Ansatz for $\mbf=(1,\ldots,1)$]\label{conj:BAm=111}
The $\Fcal_M^{(1\ldots1)}$ component of the eigenfunction of $\Hbf_2$ can be chosen
for $\xbf\in\Wcal_{1\ldots1}$ as
\beq\label{eq:BA111}
    f_M\begin{pmatrix} 1 & \ldots & 1 \\ x_1 & \ldots & x_M \end{pmatrix}=
    \sum_{\sfrak\in\Sfrak_{[1,M]}} \sgn(\sfrak)
    \left(\prod_{j<k} P(v_{\sfrak_k}-v_{\sfrak_j})\right)
    \sfrak\bigl(\exp(\vbf\cdot\xbf)\bigr),
\eeq
where $\sgn(\sfrak)$ is the sign of the permutation $\sfrak$, and
the polynomial $P(u)$ is given by \Ref{eq:defP(u)}.
\end{conj}

By construction, the Bethe wave function \Ref{eq:BA111} is antisymmetric in the momenta $\vbf$.
Note that the ratio of the coefficients for two plane waves in \Ref{eq:BA111} differing by a transposition of two adjacent momenta
$v_{\mathfrak{s}_j}$ and $v_{\mathfrak{s}_{j+1}}$ is $S(v_{\mathfrak{s}_j}-v_{\mathfrak{s}_{j+1}})$ due to \Ref{def S(u)}, as expected.

Two more conventional wave functions $f^{(\text{in})}$ and $f^{(\text{out})}$
having unitary factors at the plane waves are defined from
\beq
   f_M\begin{pmatrix} 1 & \ldots & 1 \\ x_1 & \ldots & x_M \end{pmatrix}
   =f^{(\text{in})}\,(-1)^{M(M-1)/2}\,\prod_{j<k} P(v_j-v_k)
   =f^{(\text{out})}\,\prod_{j<k} P(v_k-v_j).
\eeq

Assuming that $u_1<\ldots<u_M$ and $x_1<\ldots<x_M$, one can interpret
$\tfrak\bigl(\exp(\ri\ubf\cdot\xbf)\bigr)$,
with $\tfrak\equiv(M,\ldots,1)$,
as the incident wave, and
$\exp(\ri\ubf\cdot\xbf)$ as the outgoing scattered wave, corresponding to the ordering of the particles carrying
the momenta $u_i$ as $t\rightarrow-\infty$ or, respectively, $t\rightarrow+\infty$.

The function $f^{(\text{in})}$ is then normalised by the unit coefficient at
the incoming wave $\tfrak\bigl(\exp(\ri\ubf\cdot\xbf)\bigr)$,
and $f^{(\text{out})}$ by the unit coefficient at the outgoing wave $\exp(\ri\ubf\cdot\xbf)$.

The functions $f^{(\text{in})}$ and $f^{(\text{out})}$ differ only
by the factor (multiparticle $S$-matrix)
\beq
    f^{(\text{out})}=\Sbbd f^{(\text{in})},\qquad
    \Sbbd=\prod_{1\leq j<k\leq M} S(u_j-u_k)
\eeq
that is factorised into a product of the factors corresponding to all two-particle collisions,
in the spirit of Bethe Ansatz.

By antisymmetry in $\vbf$, the whole wave function \Ref{eq:BA111} can be restored from a single
term containing $\exp(\vbf\cdot\xbf)$.
Let $1\leq a<b\leq M$. For a subsegment $(a,\ldots,b)\subset(1,\ldots,M)$
define the polynomial
\beq\label{eq:defPbbd}
    \Pbbd_{[a,b]}(\vbf)\equiv \prod_{a\leq j<k\leq b} P(v_k-v_j)
\eeq
and the linear operator $\Pfrak_{[a,b]}$
acting on functions of $\vbf=(v_1,\ldots,v_M)$ by antisymmetrisation, with the weight factor $\Pbbd_{[a,b]}$,
in respect to  the group $\Sfrak_{[a,b]}$ of permutations of $(a,\ldots,b)$ acting on
$(v_a,\ldots,v_b)$

\beq\label{eq:defPoperator}
    \Pfrak_{[a,b]}:\ g(\vbf)\mapsto
    \sum_{\sfrak\in\Sfrak_{[a,b]}} \sgn(\sfrak)\;
    \Pbbd_{[a,b]}(\sfrak(\vbf))\, g(\sfrak(\vbf)).
\eeq

The $\Pfrak_{[a,b]}$ operator for a subsegment of $(1,\ldots,M)$ will be used only in the Appendices.
In the main text we use the abbreviation $\Pfrak\equiv\Pfrak_{[1,M]}$.

In terms of the operator $\Pfrak$, the formula \Ref{eq:BA111} for the Bethe function
simplifies to
\beq\label{eq:f111average}
   f_M\begin{pmatrix} 1 & \ldots & 1 \\ x_1 & \ldots & x_M \end{pmatrix} =
    \Pfrak\bigl(\exp(\vbf\cdot\xbf)\bigr).
\eeq

\section{Bethe Ansatz in generic sector}\label{sec:BAgeneric}

The jump conditions \Ref{eq:jump-gen} can be viewed as recurrence relations
allowing one to obtain the wave function $f_N(\xibf)$ by differentiating the wave functions $f_{N+1}(\xibf)$
corresponding to the preceding compositions (in the sense of the relation $\succ$).
Thus, starting from the source $\mbf=(1,\ldots,1)$ and travelling along the arrows of the composition graph
one can in principle obtain the wave functions for all the remaining compositions of $M$.
The problem is, however, that different paths produce, in principle, different expressions,
and one ends with a bunch of \textit{consistency conditions} for the wave function.
To show that the Bethe Ansatz works at all one has to prove that those conditions
have a joint solution.
Besides, there remain the bulk conditions \Ref{eq:bulk-gen} that also have to be verified.
The differentiation $\dd_{x_i}$ acting on the exponent in \Ref{eq:f111average}
are replaced by $v_i$, and the resulting consistency equations take the form of a (very overdetermined)
set of algebraic equations for the coefficients of the Bethe wave functions.

To any composition $\mbf=(m_1,\ldots,m_N)$ of $M$ of length $\ell(\mbf)=N$,
there corresponds a split of the sequence of the momenta $v_i$
\beq
     \vbf=(v_1,\ldots,v_M)=(\wbf_1^{\mbf};\ldots;\wbf_N^{\mbf})
\eeq
into the consecutive segments $\wbf_j^{\mbf}$ of respective length $m_j$, so that
\beq
     \wbf_j^{\mbf} = (v_{m_1+\ldots+m_{j-1}+1},\ldots,v_{m_1+\ldots+m_j}),\qquad
     j=1,\ldots,N
\label{definition vecteur omega j comp m}
\eeq
or
\beq
     (\wbf_j^{\mbf})_i = v_{m_1+\ldots+m_{j-1}+i},\qquad i = 1,\ldots,m_j.
\eeq

At the vertex $\mbf$, the original coordinates $(x_1,\ldots,x_M)$ merge into consecutive groups
of length $m_i$:
\beq
     (x_1,\ldots,x_M)\mapsto
     \Xbf^{\mbf}=(\overbrace{x_1,\ldots,x_1}^{m_1},\overbrace{x_2,\ldots,x_2}^{m_2},\ldots,\overbrace{x_N,\ldots,x_N}^{m_N}),
\eeq
and now we set $\xbf=(x_1,\ldots,x_N)$.

Let $\angles{\wbf}$ denote the sum of the components of a vector $\wbf$, e.g.\
$\angles{\vbf}=v_1+\ldots+v_M$. Let also
\beq
\Wbf^{\mbf} = (\angles{\wbf_1^{\mbf}},\ldots,\angles{\wbf_N^{\mbf}})
\label{definition vecteur W comp m}
\eeq
so that
\beq
     \Wbf^{\mbf}\cdot\xbf =\angles{\wbf_1^{\mbf}}x_1+\ldots+\angles{\wbf_N^{\mbf}}x_N,
\eeq
and
\beq
    \vbf\cdot\Xbf^{\mbf} = \Wbf^{\mbf}\cdot\xbf.  
\eeq

\begin{conj}[Bethe Ansatz for generic $\mbf$]\label{conj:BAgeneric}
The Bethe eigenfunction in the generic subsector $\Fcal_M^{\,\mbf}$
can be written in the form
\beq\label{eq:BAgeneric}
      f_N\begin{pmatrix}
             m_1 &\ldots& m_N\\
             x_1 &\ldots& x_N
            \end{pmatrix}
      = \Pfrak\bigl(Q^{\mbf}(\vbf)\,\exp(\Wbf^{\mbf}\cdot\xbf)\bigr),
\eeq
where $Q^{\mbf}(\vbf)$ is a polynomial in $\vbf$.
In particular, $Q^{(1\ldots1)}(\vbf)\equiv1$.
\end{conj}

If $\Pfrak:g\mapsto0$ for some function $g(\vbf)$ we shall say that $g(\vbf)$ is \textit{$\Pfrak$-reducible}
and write $g\Pequiv0$. If $g_1-g_2\Pequiv0$  we shall say that
$g_1$ and $g_2$ are \textit{$\Pfrak$-equivalent} and write $g_1\Pequiv g_2$.
As a consequence, all quantities under the sign of
$\Pfrak$ are defined only up to $\Pfrak$-equivalence.

Consider the bulk equation \Ref{eq:bulk-gen}.

Since the vertex $\mbf=(1,\ldots,1)$ is the source of the composition graph, \textit{c.f.}\ Fig.\ \ref{fig:hypercubes},
having no predecessors,
the corresponding bulk equation \Ref{eq:bulk-gen} contains no $\b$-terms and is obviously
satisfied by the Ansatz \Ref{eq:BA111} producing the eigenvalue $\la$ of $\Hbf_2$
\beq
     \la = -v_1^2-\ldots-v_M^2 +M\g_1.
\eeq

Let $K(\wbf)$ be the sum of the squares of the components of a vector $\wbf$.
Recalling that $\g_1=\g$, by \Ref{def-betam1m2gammam},
we have $\la = -K(\vbf)+M\g_1$ or, splitting the sum into groups of size $m_j$,
\beq
    \la = \sum_{j=1}^N \bigl(m_j\g-K(\wbf_j^{\mbf})\bigr).
\eeq

Each $\dd_{x_j}$ in \Ref{eq:bulk-gen} is replaced by $\angles{\wbf_j^{\mbf}}$.
Moving the $\la$ term into the right-hand-side we find that the coefficient in front of $f_N(\xibf)$ produces the factor
\beq
     \sum_{j=1}^N \widetilde{K}(\wbf_j^{\mbf}) +(m_j^3-m_j)\g,
\eeq
where we use the notation
\beq
     \widetilde{K}(\wbf_j^{\mbf}) = K(\wbf_j^{\mbf}) - \frac{\angles{\wbf_j^{\mbf}}^2}{m_j}
\eeq
for the ``kinetic energy'' of the cluster $\wbf_j^{\mbf}$ reduced w.r.t.\ the center-of-mass.
Note that $\widetilde{K}(\wbf_j^{\mbf})$ is invariant under translations
\mbox{$\wbf_j^{\mbf}\mapsto\wbf_j^{\mbf}+(c,\ldots,c)$},
as a manifestation of the Galilei invariance.

Upon summing up over the compositions  $\mbf'\succ\mbf$ introduced in \eqref{compositions for bulk equation}, the bulk equation \Ref{eq:bulk-gen} takes finally the form
\beq\label{eq:bulk-gen-Q} 
    \left[\left(\sum_{j=1}^N \widetilde{K}(\wbf_j^{\mbf}) +(m_j^3-m_j)\g\right)\,Q^{\mbf}(\vbf)
    +\sum_{\mbf'\succ\mbf} \b_{n_1n_2} Q^{\mbf'}(\vbf)\right]
    \exp(\Wbf^{\mbf}\cdot\xbf)
      \Pequiv0.
\eeq

After performing differentiations of the exponent, the jump equation \Ref{eq:jump-gen}
can be recast in terms of the compositions $\mbf', \mbf''$ of $\mbf$ introduced in \eqref{introduction compositions eqn saut}
\begin{multline}\label{eq:jump-gen-Q} 
\biggl[ V_{m_km_{k+1}}\big(\wbf_k^{\mbf}\big) \cdot Q^{\mbf'}(\vbf)
    +V_{m_{k+1}m_{k}}\big(\wbf_k^{\mbf}\big) \cdot Q^{\mbf''}(\vbf) \\
    +2m_km_{k+1}\b_{m_k,m_{k+1}} Q^{\mbf}(\vbf)
      \biggr]\exp(\Wbf^{\mbf}\cdot\xbf) \Pequiv0.
\end{multline}
Here, given an $n_1+n_2$-dimensional vector $\boldsymbol{u}=(u_1,\dots,u_{n_1+n_2})$, we have defined
\begin{equation}
    V_{n_1n_2}(\boldsymbol{u}) \, = \, n_2(u_1+\ldots+u_{n_1})-n_1(u_{n_1+1}+\ldots+u_{n_1+n_2}) \;.
\end{equation}

\begin{conj}[Factorisation property]\label{conj:Qfactor}
The polynomial $Q^{\mbf}(\vbf)$ in \Ref{eq:BAgeneric}
can be chosen, up to a $\Pfrak$-equivalent expression, in the factorised form
\beq\label{eq:Qfactor}
     Q^{\mbf}(\vbf) = \prod_{k=1}^N Q^{(m_k)}(\wbf_k^{\mbf}).
\eeq
\end{conj}

A justification of the above conjecture is presented in Appendix \ref{app:factorisation}.

\section{Bethe Ansatz in sector $\Fcal_M^{(M)}$}\label{sec:BAm=(M)}

Conjecture \ref{conj:Qfactor} suggests that it is
sufficient to analyze the consistency equations only for the compositions of unit length $\mbf=(M)$,
$N=1$. In this case,
the exponent
$$\exp(\Wbf^{\mbf}\cdot\xbf)=\exp((v_1+\ldots+v_M)x_1)$$
becomes completely symmetric and can be factored out from under $\Pfrak$.
The equations for $Q^{(M)}(\vbf)$ are thus purely polynomial.

For the bulk equation \Ref{eq:bulk-gen-Q} we have now
\be
     \mbf'=(n_1,n_2),\quad n_1+n_2=M,\quad \b_{n_1n_2}=\sqrt{n_1n_2M}\, \b
\ee
\be
    \wbf_1^{\mbf'}=(v_1,\ldots,v_m),\qquad \wbf_2^{\mbf'}=(v_{m+1},\ldots,v_M),
\ee
and the equation takes form
\begin{multline}\label{eq:bulk-QM} 
    \bigl(\widetilde{K}(\vbf)+(M^3-M)\g\bigr)\,Q^{(M)}(v_1,\ldots,v_M)\\
      +\b\sum_{\substack{n_1,n_2\geq1\\ n_1+n_2=M}} \sqrt{n_1n_2M}\,
      Q^{(n_1)}(v_1,\ldots,v_{n_1})\,Q^{(n_2)}(v_{n_1+1},\ldots,v_M)
      \Pequiv0.
\end{multline}

For the jump equation \Ref{eq:jump-gen-Q} we consider the compositions
\be
   \mbf'=(m_1,m_2),\qquad \mbf''=(m_2,m_1),\qquad \mbf=(m_1+m_2),
\ee
and denote $\vbf=(v_1,\dots, v_{m_1+m_2})$, what recasts the equation in the form

\begin{multline}\label{eq:jump-QM} 
       V_{m_1m_2} (\vbf)  \,
       Q^{(m_1)}(v_1,\ldots,v_{m_1})\,Q^{(m_2)}(v_{m_1+1},\ldots,v_{m_1+m_2}) \\
       +V_{m_2m_1} (\vbf) \,
      Q^{(m_2)}(v_1,\ldots,v_{m_2})\,Q^{(m_1)}(v_{m_2+1},\ldots,v_{m_1+m_2}) \\
     + 2m_1m_2\b_{m_1m_2}\,Q^{(m_1+m_2)}(v_1,\ldots,v_{m_1+m_2}) \Pequiv0.
\end{multline}

The equations \Ref{eq:bulk-QM} and \Ref{eq:jump-QM}
together with $Q^{(1)}(v)=1$
constitute
the complete set of conditions for the polynomials $Q^{(M)}(v_1,\ldots,v_M)$, $M=1,2,\ldots$,
defined up to $\Pfrak$-equivalence.
Extensive computer experiments have led us to the following explicit solution
to the equations \Ref{eq:bulk-QM} and \Ref{eq:jump-QM}.

\begin{conj}[Solution]\label{conj:Qsolution}
Set $Q^{(1)}(v)=1$ and for $M\geq2$ define
the polynomial $Q^{(M)}(\vbf)$
as the following homogeneous polynomials of degree $M-1$
\beq\label{eq:Qsolution} 
      Q^{(M)}(\vbf) = \frac{2\sqrt{M}}{M!(M-1)}(2\b)^{1-M}\,
      \sum_{1\leq i<j\leq M} (-1)^{j-i}\binom{M-1}{j-i-1}(v_i-v_j)^{M-1}
\eeq
invariant under translations $v_i\mapsto v_i+c$.
Then such $Q^{(M)}(\vbf)$
satisfy all the equations \Ref{eq:bulk-QM} and \Ref{eq:jump-QM}.
\end{conj}

Note that $Q^{(M)}(\vbf)$ do not contain coupling constants $\b$, $\g$
that are hidden inside the $\Pfrak$-operator.

Conjecture \ref{conj:Qsolution}
has been confirmed by means of computer algebra for $M\leq8$.
In fact, instead of verifying Conjectures \ref{conj:Qsolution} literally,
we have verified a stronger Conjecture \ref{conj:Qeqs23red},
see  Appendix \ref{app:2-3reducibility}.

\section{Discussion}\label{sec:discussion}

As a test of quantum integrability of the system, we have demonstrated consistency
of the Bethe Ansatz for $M\leq8$. This is a pretty convincing though not conclusive result.
A rigorous proof of Conjecture \ref{conj:Qsolution}, or superseding Conjecture \ref{conj:Qeqs23red}
remains an open problem. The $\Pfrak$ operator, and the notions of 2- and 3-reducibility introduced
in Appendix \ref{app:2-3reducibility} seem to be new combinatorial objects that might be of interest
for themselves.

An alternative way to establish quantum integrability could be provided through the
Algebraic Bethe Ansatz \cite{BogoliubiovIzerginKorepinBookCorrFctAndABA}
based on quantum Lax operator and $R$-matrix. That would also help to identify the underlying quantum algebra.
The work in this direction is in progress.

Except for $M=2$, we have not pursued a comprehensive study of the orthogonality and completeness
of the Bethe eigenfunctions, neither of the structure of bound states.
In the case of the quantum nonlinear Schr\"odinger equation (delta-function Bose gas)
it is known that the bound states of the quantum model correspond in the classical limit to the
solitons of the classical model \cite{KulishManakovFaddeev1976}. It would be interesting to study a similar
correspondence for the KP-model.

The model we study is associated with a cubic polynomial $P$ with zero sum of the roots \Ref{eq:def-Pv}
through which the two-particle $S$-matrix is expressed \Ref{def S(u)} and, in turn, the factorised
multiparticle $S$-matrix. The question arises what possible QFT models could be associated
with polynomials $P$ of higher degree, or without the restriction on the roots.
In \cite{KozlowskiSklyanin2013} the properties of Bethe equations associated with a
generic polynomial $P$ were studied in an abstract way, without clarifying the nature
of the corresponding QFT. In a recent paper \cite{Litvinov2013} a possible example of a model of that class
is proposed.

The model we study is nonrelativistic and Galilei invariant. It appears that it corresponds to
a nonrelativistic limit of a relativistic integrable model known as
affine $A_{N-1}$ Toda field theory \cite{arinshtein1979quantum,braden1990affine}
and given by the Lagrangian
\beq
 \Lcal = \frac12 \sum_{i=1}^N \dd_\mu\phi_i\cdot\dd^\mu\phi_i
 -\frac{2M^2}{\b^2} \sum_{i=1}^N
 \exp \left[ \frac{\b}{\sqrt{2}} (\phi_i-\phi_{i+1} ) \right].
\eeq

Indeed, the $S$-matrix for a pair of main particles of the Toda FT
is conjectured in \cite{arinshtein1979quantum} to be
\beq
  S_{11}(\theta) = \frac{ \sinh\Big( \frac{\theta}{2} + \frac{ \ri \pi }{ N } \Big)
                   \sinh\Big( \frac{\theta}{2} - \frac{ \ri \pi }{ N } + \ri\frac{b}{2}\Big)
                   \sinh\Big( \frac{\theta}{2} -  \ri\frac{b}{2}\Big) }
         {  \sinh\Big( \frac{\theta}{2} - \frac{ \ri \pi }{ N } \Big)
            \sinh\Big( \frac{\theta}{2} + \frac{ \ri \pi }{ N } - \ri\frac{b}{2}\Big)
            \sinh\Big( \frac{\theta}{2}  +  \ri\frac{b}{2}\Big)  }.
\eeq

Upon carrying out the rescaling
\beq
    \theta = \frac{2 \kappa^{-1} \pi  u }{ N } \qquad \text{and}  \qquad b = \frac{2 \kappa^{-1} \tau \pi  }{ N }
\eeq
and then sending $N\rightarrow +\infty$ one obtains the rational degeneration
\beq
       \lim_{N\rightarrow +\infty} S_{11}(\theta) =  \widetilde{S}_{11}(u) =
       \frac{ u^3 +  u \big( \kappa^2 + \tau^2 - \kappa\tau \big) - \ri  \kappa\tau(\kappa-\tau)  }
       {  u^3 +  u \big( \kappa^2 + \tau^2 - \kappa\tau \big) + \ri  \kappa\tau(\kappa-\tau)    }.
\eeq

Thus choosing $\tau$ and $\kappa$ such that $\kappa^2 + \tau^2 - \kappa\tau = -12 \g$ and $ \kappa\tau(\kappa-\tau)=4\b^2$ one obtains the scalar $S$-matrix of qKP, or more precisely,
qKP-II, since $0<b<2\pi/N$:
\beq
   \widetilde{S}_{11}(u)= S_{\text{qKP}}(u)=  \frac{ u^3-12 \g u -4\ri \b^2 }{ u^3-12 \g u  + 4\ri \b^2 }.
\eeq

\textbf{Acknowledgements:} We thank Beno\^{i}t Vicedo for valuable discussions at the initial stages of the project.
A.T.\ thanks the EPSRC for funding under the First Grant project EP/K014412/1, and the STFC for support under the Consolidated Grant project nr. ST/L000490/1.

\appendix
\section{}\label{app:factorisation}

In the appendices we shall use again the non-abbreviated notation $\Pbbd_{[a,b]}$ \Ref{eq:defPbbd}
and $\Pfrak_{[a,b]}$ \Ref{eq:defPoperator} for a subsegment $(a,\ldots,b)\subset(1,\ldots,M)$.

Let us state a couple of elementary properties of the operator $\Pfrak_{[a,b]}$.
It is assumed below that $\vbf=(v_1,\ldots,v_M)$.

\begin{lemma}\label{lemma:symmfactor}
If a function $F(\vbf)$ is $\Pfrak_{[a,b]}$-reducible, and
a function $G(\vbf)$ is symmetric under permutations $\Sfrak_{[a,b]}\subset\Sfrak_{[1,M]}$
then the product $F(\vbf)G(\vbf)$ is also $\Pfrak_{[a,b]}$-reducible.
\end{lemma}

\begin{proof}
Since  $G(\vbf)$ is invariant under $\Sfrak_{[a,b]}$ it is factored out from the sum over $\Sfrak_{[a,b]}$
in \Ref{eq:defPoperator}.
\end{proof}

\begin{lemma}\label{lemma:subsegmentPreduc}
If a function $F(\vbf)$ is $\Pfrak_{[a,b]}$-reducible then $F(\vbf)$ is also $\Pfrak_{[1,M]}$-reducible.
\end{lemma}

\begin{proof}
The product $\Pbbd_{[1,M]}$ factorises as
$\Pbbd_{[1,M]} = \Pbbd_{[a,b]}\cdot \widetilde\Pbbd$ where the complementary factor
$\widetilde\Pbbd$ is  $\Sfrak_{[a,b]}$-symmetric,
hence the product $F(\vbf)\widetilde\Pbbd$ is $\Pfrak_{[a,b]}$-reducible,
by Lemma \ref{lemma:symmfactor}.

The sum over the group $\Sfrak_{[1,M]}$ can be rewritten as the double sum,
first over the subgroup $\Sfrak_{[a,b]}\subset\Sfrak_{[1,M]}$, then over the coset
$\Sfrak_{[1,M]}/\Sfrak_{[a,b]}$. The alternating sum over $\Sfrak_{[a,b]}$ with the weight
$\Pbbd_{[a,b]}$ then nullifies $F(\vbf)\widetilde\Pbbd$.
\end{proof}

Now we can justify  Conjecture \ref{conj:Qfactor}.

\begin{prop}\label{prop:factorisation}
Assume that a sequence of polynomials $Q^{(m)}(v_1,\ldots,v_m)$, $m=1,2,\ldots$ solve equations
\Ref{eq:bulk-QM} and \Ref{eq:jump-QM}. Then $Q^{\mbf}(\vbf)$ given by \Ref{eq:Qfactor}
solve equations \Ref{eq:bulk-gen-Q} and \Ref{eq:jump-gen-Q}.
\end{prop}

\begin{proof}

Given the composition $\mbf = (m_1,\ldots,m_{k-1},m_k,m_{k+1},\ldots,m_N)$, let $(a_i,\ldots,b_i)$ be the consecutive subsegments of length $m_i$ of the sequence $(1,\ldots,M)$,
or, explicitly,
\beq
     a_i=m_1+\ldots+m_{i-1}+1,\qquad b_i=m_1+\ldots+m_i,\qquad i=1,\ldots,N.
\eeq
By using the product structure of $Q^{\mbf}$ one can recast the below sum as
\begin{multline}
    \left[\left(\sum_{j=1}^N \widetilde{K}(\wbf_j^{\mbf}) +(m_j^3-m_j)\g\right)\,Q^{\mbf}(\vbf)
    +\sum_{\mbf':\ \mbf'\succ\mbf}  \b_{n_1n_2} Q^{\mbf'}(\vbf)\right]
    \text{e}^{\Wbf^{\mbf}\cdot\xbf} \\
\; = \; \sum_{j=1}^{N} \Bigg\{ \Big[ \widetilde{K}(\wbf_j^{\mbf}) +(m_j^3-m_j)\g \Big]  Q^{(m_j)}(\wbf_j^{\mbf})  \\
\; + \; \sum_{\substack{n_1+n_2 \\ =m_j}}\b_{n_1n_2} Q^{(n_1)}\big(w_{j,1}^{\mbf}, \dots, w_{j,n_1}^{\mbf} \big) Q^{(n_2)}\big(w_{j,n_1+1}^{\mbf}, \dots, w_{j,m_j}^{\mbf} \big)
\Bigg\} \cdot\text{e}^{\Wbf^{\mbf}\cdot\xbf} \cdot \prod_{\substack{a=1 \\ \not=j } }^{N}Q^{(m_j)}(\wbf_a^{\mbf}) \;.
\nonumber
\end{multline}
Here $\xbf=(x_1,\ldots,x_N)$, $\Wbf^{\mbf}$ are as defined  in \eqref{definition vecteur W comp m} while the vector $\wbf_j^{\mbf}$ introduced in \eqref{definition vecteur omega j comp m}
have components
$$
\wbf_j^{\mbf} \, = \, \big(w_{j,1}^{\mbf}, \dots, w_{j,m_j}^{\mbf} \big) \;.
$$
The product outside of the bracket is symmetric in respect to permutations of the coordinates of $\wbf_j^{\mbf}$, \textit{viz}. in respect to the action of the permutation group $\mathfrak{S}_{[a_j,b_j]}$.
By construction, the functions appearing inside of the brackets are $\Pfrak_{[a_j,b_j]}$-reducible.  Thus $Q^{\mbf}(\vbf)$ given by \Ref{eq:Qfactor}
solves the bulk equation \Ref{eq:bulk-gen-Q} in virtue of Lemma \ref{lemma:symmfactor}.

It thus remains to deal with the gluing conditions issuing from the jump of the transversal derivatives.
Here, we introduce the auxiliary compositions $\mbf'\succ\mbf$, $\mbf''\succ\mbf$
%
%
\begin{align*}
\mbf'&=(m_1,\ldots,m_{k-1},n_1,n_{2},m_{k+1},\ldots,m_N),\\
    \mbf''&=(m_1,\ldots,m_{k-1},n_{2},n_1,m_{k+2},\ldots,m_N)
\end{align*}
with $n_1+n_2=m_k$. Then, it holds
\begin{multline}\label{eq:jump-gen-Q} 
    \biggl[ V_{n_1n_{2}}\big(\wbf_k^{\mbf}\big) \cdot Q^{\mbf'}(\vbf)
    +V_{n_{2}n_{1}}\big(\wbf_k^{\mbf}\big) \cdot Q^{\mbf''}(\vbf)  + 2n_1 n_{2}\b_{n_1,n_{2}} Q^{\mbf}(\vbf)  \biggr]\cdot\text{e}^{\Wbf^{\mbf}\cdot\xbf} \\
 = \biggl[2n_1 n_{2}\b_{n_1,n_{2}} Q^{(m_k)}(\wbf_k^{\mbf})  \, + \,  V_{n_1n_{2}}\big(\wbf_k^{\mbf}\big) \cdot Q^{(n_1)}\big(w_{k,1}^{\mbf},\dots, w_{k,n_1}^{\mbf}\big) Q^{(n_2)}\big(w_{k,n_1+1}^{\mbf},\dots, w_{k,m_k}^{\mbf}\big) \\
 +V_{n_{2}n_{1}}\big(\wbf_k^{\mbf}\big) \cdot Q^{(n_2)}\big(w_{k,1}^{\mbf},\dots, w_{k,n_2}^{\mbf}\big) Q^{(n_1)}\big(w_{k,n_2+1}^{\mbf},\dots, w_{k,m_k}^{\mbf}\big)   \biggr]
 \cdot\text{e}^{\Wbf^{\mbf}\cdot\xbf} \cdot \prod_{\substack{a=1 \\ \not=k } }^{N}Q^{(m_j)}(\wbf_a^{\mbf}) \;.
\nonumber
\end{multline}
Again, the product outside of the bracket is symmetric in respect to permutations of the coordinates of $\wbf_k^{\mbf}$, \textit{viz}. $\mathfrak{S}_{[a_j,b_j]}$.
By construction, the functions appearing inside of the brackets are $\Pfrak_{[a_j,b_j]}$-reducible.  Thus $Q^{\mbf}(\vbf)$ given by \Ref{eq:Qfactor}
solves the jump of transversal derivative condition \Ref{eq:jump-gen-Q} in virtue of Lemma \ref{lemma:symmfactor}.

\end{proof}

\section{}\label{app:2-3reducibility}

Checking $\Pfrak$-equivalence of polynomials directly is difficult even with computer
since it involves summation over $M!$ permutations, which leads to the exponential growth of
the computational complexity with $M$.
When verifying Conjecture \ref{conj:Qsolution}, we checked in fact some
stronger conditions
that we call  $2$- and $3$-reducibility having the advantage of a polynomial complexity.

Let $P(v)$ be given by \Ref{eq:defP(u)} and  $\vbf=(v_1,\ldots,v_M)$.
Let $v_{ij}=v_i-v_j$, and $P_{ij}=P(v_i-v_j)$.
Assuming $M\geq2$, we shall say that a polynomial $F(\vbf)$ is \textit{$2$-reducible} and write $F\twored0$ if $F(\vbf)$
admits a decomposition
\beq\label{eq:def2reduction}
     F(\vbf) = \sum_{i=1}^{M-1} P_{i,i+1}\,G_i(\vbf)
\eeq
with some polynomials $G_i(\vbf)$ such that $G_i(\vbf)$ is symmetric under permutation $v_i\leftrightarrow v_{i+1}$
for each $i$.
Note that such a decomposition is not necessarily unique.

\begin{prop}\label{prop:2reducibility}
If $F$ is 2-reducible then $F$ is $\Pfrak_{[1,M]}$-reducible.
\end{prop}

\begin{proof}
Note that $P_{i,i+1}$ is $\Pfrak_{[i,i+1]}$-reducible since
$P_{i,i+1}\Pbbd_{[i,i+1]}=P_{i,i+1}P_{i+1,i}$
is $\Sfrak_{[i,i+1]}$-symmetric, hence nullified by the antisymmetrisation.
Then, by Lemma \ref{lemma:symmfactor}, the $i$-th term in \Ref{eq:def2reduction}
is $\Pfrak_{[i,i+1]}$-reducible, hence $\Pfrak_{[1,M]}$-reducible, by Lemma \ref{lemma:subsegmentPreduc}.
\end{proof}

The property of $2$-reducibility is not always sufficient to prove the $\Pfrak$-reducibility,
and we shall also use the notion of \textit{$3$-reducibility} defined below.

\begin{lemma}\label{lemma:v12-v23}
The polynomial $v_{12}-v_{23}=v_1-2v_2+v_3$ is $\Pfrak_{[1,3]}$-reducible.
\end{lemma}

\begin{proof}
Note that $P_{12}$ is $\Pfrak_{[1,2]}$-reducible, and $P_{23}$ is $\Pfrak_{[2,3]}$-reducible,
as shown in the proof of Proposition \ref{prop:2reducibility}.
By Lemma \ref{lemma:subsegmentPreduc}, $P_{12}$, $P_{23}$, and therefore $P_{12}-P_{23}$
are $\Pfrak_{[1,3]}$-reducible.
Now note that the difference
\be
     P_{12}-P_{23} = v_{12}^3-v_{23}^3 + 12\g(v_{12}-v_{23})
\ee
factorises into $v_{12}-v_{23}$ and a quadratic polynomial $J$ that is
$\Sfrak_{[1,3]}$-symmetric:
\beq
     P_{12}-P_{23} = (v_{12}-v_{23})J ,
\eeq
\beq
     J  = v_{12}^2+v_{12}v_{23}+v_{23}^2+12\g = v_1^2+v_2^2+v_3^2-v_1v_2-v_1v_3-v_2v_3+12\g.
\eeq

Then from the symmetry of $J $ it follows that
\beq
    0=\Pfrak_{[1,3]}(P_{12}-P_{23})=\Pfrak_{[1,3]} \bigl((v_{12}-v_{23})J \bigr) = J \,\Pfrak_{[1,3]} (v_{12}-v_{23})
\eeq
and therefore $\Pfrak_{[1,3]} (v_{12}-v_{23}) =0$ since $J \neq0$.
\end{proof}

By Lemma \ref{lemma:subsegmentPreduc}, an immediate corollary is that
$v_{i,i+1}-v_{i+1,i+2}$ is $\Pfrak_{[1,M]}$-reducible for any $M$,
and $i=1,\ldots,M-2$.

Remarkably, the condition that
\beq
      J  = \frac{P_{12}-P_{23}}{v_{12}-v_{23}}
\eeq
is an $\Sfrak_{[1,3]}$-symmetric polynomial
fixes the polynomial $P(v)$ uniquely as a cubic polynomial with zero $v^2$-term.
The easiest way to prove this is to use the homogeneity and to check the monomials $v^p$ to see
that the solution is $p\in\{0,1,3\}$.

Assuming $M\geq3$, we shall say that a polynomial $F(\vbf)$ is \textit{$3$-reducible} and write $F\threered0$ if $F(\vbf)$
admits a decomposition
\beq\label{eq:def3reduction}
     F(\vbf) = \sum_{i=1}^{M-2} (v_{i,i+1}-v_{i+1,i+2})\,J_i(\vbf)
\eeq
with some $\Sfrak_{[i,i+2]}$-symmetric polynomials $J_i(\vbf)$.

Note that such a decomposition is not necessarily unique.

\begin{prop}\label{prop:3reducibility}
If $F$ is 3-reducible then $F$ is $\Pfrak_{[1,M]}$-reducible.
\end{prop}

\begin{proof}
For the $i$-th term in \Ref{eq:def3reduction} we have
$$\Pfrak_{[i,i+2]}\big( (v_{i,i+1}-v_{i+1,i+2})\,J_i(\vbf) \big) \, = \,  J_i(\vbf)\Pfrak_{[i,i+2]}(v_{i,i+1}-v_{i+1,i+2})=0,$$
using first the symmetry of $J_i$, then Lemma \ref{lemma:v12-v23}. By Lemma \Ref{lemma:subsegmentPreduc},
each term is $\Pfrak_{[1,M]}$-reducible.
\end{proof}

The following conjecture supersedes Conjecture \ref{conj:Qsolution}.
It has been verified by means of computer algebra for $M\leq8$.

\begin{conj}\label{conj:Qeqs23red}
For the polynomials $Q^{(M)}(\vbf)$ given by \Ref{eq:Qsolution}
the left-hand-side of the jump equation \Ref{eq:jump-QM}
is in fact $3$-reducible, which, by Proposition \ref{prop:3reducibility},
implies $\Pfrak_{[1,M]}$-reducibility.

The  left-hand-side of the bulk equation \Ref{eq:bulk-QM}
is respectively a sum of a 2-reducible and a 3-reducible parts, which,
by Propositions \ref{prop:2reducibility} and \ref{prop:3reducibility},
implies $\Pfrak_{[1,M]}$-reducibility.
\end{conj}


\end{document}